\documentclass[11pt,oneside]{article}
\usepackage[square, numbers]{natbib}
\usepackage{amsmath, amssymb}
\usepackage{amsthm,xcolor}
\usepackage{mathtools}
\usepackage{amsfonts}
\usepackage{hyperref}
\usepackage{enumitem}
\usepackage{bm, bbm, cancel}
\usepackage{hhline}             

\newcommand{\R}{\mathbb{R}}
\newcommand{\N}{\mathbb{N}}

\def\d{{\rm d}}

\numberwithin{equation}{section}

\newtheorem{theorem}{Theorem}[section] 
\newtheorem{lemma}[theorem]{Lemma}
\newtheorem{corollary}[theorem]{Corollary}
\newtheorem{proposition}[theorem]{Proposition}

\theoremstyle{definition}
\newtheorem{definition}[theorem]{Definition}

\newtheorem{remark}[theorem]{Remark}

\newtheorem{examplex}[theorem]{Example}
\newenvironment{example}
{\pushQED{\qed}\examplex}
{\popQED\endexamplex}

\title{Pathwise Portfolio Theory and Market Viability}

\author{
	\textsc{Ioannis Karatzas}
	\thanks{Departments of Mathematics and Statistics, Columbia University, USA (E-mail: {\it ik1@columbia.edu})}
	\and
	\textsc{Donghan Kim} 
	\thanks{Department of Mathematical Sciences, KAIST, South Korea (E-mail: {\it kimdonghan@kaist.ac.kr})}
}

\date{\today}

\begin{document}

\maketitle

\begin{abstract}
    \noindent The theory of portfolios, and its allied notions and fundamental results concerning growth optimality, the num\'eraire property, and ``market viability''---which rules out the possibility of financing nontrivial future liability streams starting with arbitrarily small initial capital---is developed in a pathwise setting, completely devoid of probabilistic considerations. The approach replaces the familiar semimartingale decomposition of stochastic analysis for assets' returns, by decompositions generated through suitable trend extractors and their associated residual paths; then deploys F\"ollmer's celebrated pathwise version of classical It\^o integration and calculus. The resulting growth--num\'eraire and viability--boundedness equivalences bear considerable similarities to their semimartingale counterparts, but need not collapse into a single equivalence class in the pathwise setting; this separation is illustrated by two examples.
\end{abstract}

\section{Introduction}

A central theme of continuous-time portfolio theory is that several notions which appear quite unrelated at first sight are, in fact, manifestations of the same underlying market structure. In the continuous semimartingale setting for asset returns, the existence of a distinguished num\'eraire portfolio, the finiteness of the best available growth opportunities, the absence of extreme forms of arbitrage, and the boundedness of attainable levels of wealth are closely connected. A particularly comprehensive formulation is Theorem~2.31 of \citet{KaratzasKardaras2021}, which identifies market viability with the existence of a local martingale deflator, the existence of a supermartingale (or a local martingale) num\'eraire, locally finite growth, and boundedness-in-probability of attainable levels of wealth. Together with the allied results on growth optimality, this theorem is one of the cornerstones of the semimartingale theory of portfolios and arbitrage.

The semimartingale equivalences underlying this cornerstone result have a broader lineage. General num\'eraire-portfolio theory, including characterizations in terms of predictable characteristics and boundedness-in-probability, was developed by \citet{KaratzasKardaras2007}, while \citet{Kardaras2012} related market viability to absence of arbitrage of the first kind and to local martingale deflators. The same viability-based research has since been extended to open equity markets with changing constituents \citep{Karatzas:Kim2} and to markets whose dimension changes over time \citep{BayraktarKimTilva2024}.
 
{\it We investigate here to what extent this structure survives when probability is removed altogether.} We work with continuous trajectories of cumulative asset returns, impose finite quadratic covariation only along a fixed refining sequence of partitions, and use the pathwise integration and change-of-variable calculus initiated by \citet{follmer1981}. There is no probability measure, no conditional expectation, and no martingale assumption. Our objective is to identify which parts of the classical theory are genuinely pathwise, and which rely essentially on semimartingales.

Probability-free portfolio and integration theory has developed in several complementary directions. Strictly pathwise master formulas and functional portfolio generation within F\"ollmer calculus appear in \citet{Schied:model-free} and \citet{Karatzas:Kim}; rough-path extensions allowing broader portfolio classes are developed in \citet{AllanCuchieroLiuPromel2023}. Related model-free treatments of pathwise integration, self-financing, and hedging include \citet{PerkowskiPromel2016} and \citet{ChiuCont2023}. Our question is different: rather than deriving portfolio-generation or hedging formulas, we seek a pathwise counterpart of the growth, num\'eraire, financing, and viability structure. The classical cornerstone theorem combines statements of two kinds: local structural relations among return characteristics and portfolio growth, and global restrictions on what can be financed. The former admits a pathwise formulation on an individual trajectory, whereas the latter necessarily refers to a collection of possible paths and to a single trading rule that respects the information available on each of them. Keeping these roles separate from the outset is essential to the formulation of our pathwise result.

The first difficulty is that a continuous path does not come with a canonical drift-noise decomposition. We address this by fixing an observation-window length and a non-anticipative trend extractor. At every time, the extractor reads only the most recently observed portion of the return path, and produces a trend signal at the chosen resolution. Accumulating this signal yields a trend component of finite first variation, while the remainder is treated as the residual path. Since subtraction of a continuous finite variation path does not alter F\"ollmer quadratic covariation, the residual retains the entire second-order roughness of the original return trajectory. The trend and covariation components then serve as finite-resolution pathwise characteristics. Faber--Schauder projections provide a concrete and flexible family of such extractors.

The dependence on extractor and observation scale is deliberate. In the absence of probability, there is no canonically specified drift component which the pathwise construction could hope to recover. Our characteristics describe instead what is regarded as trend and what as noise at a prescribed resolution, using only information already observed. The results should therefore be read relative to this choice. This finite-resolution viewpoint makes the role of the extracted characteristics explicit, rather than concealing it in an assumed probabilistic model.

Two additional issues arise when these characteristics are used for trading. First, the associated local rates are ordinarily determined only almost everywhere, while F\"ollmer sums evaluate a portfolio at the actual partition points. We select pointwise versions through a fixed backward-looking differentiation rule, so that the resulting objects remain measurable and non-anticipative. Secondly, not every measurable portfolio rule can be integrated against an arbitrary path. We introduce a concrete class of wealth-generating F\"ollmer portfolios, built from the standard admissible integrands of pathwise It\^o calculus and closed under finite switching at partition times. Within this class, wealth is well defined, strictly positive, and self-financing.

The first group of our results concerns growth and num\'eraires. At the level of the selected pathwise characteristics, we identify the condition under which growth opportunities remain finite, and show that it is equivalent to the existence of a portfolio satisfying the corresponding compatibility condition between trend and covariation. The same condition characterizes growth optimality in the sense determined by the extracted characteristics, and yields an exact comparison of growth for any two portfolios. When the canonical portfolio can be realized as a wealth-generating F\"ollmer portfolio, the F\"ollmer logarithm of every relative wealth process, denominated by the wealth this portfolio generates, has no extracted finite variation return drift. We call such a portfolio's wealth  \emph{driftless num\'eraire}. This terminology is deliberately algebraic: without a probability measure, the driftless relative-wealth identity does not imply a supermartingale property. We isolate the realizability of this canonical portfolio as a separate admissibility requirement and give concrete sufficient conditions for it.

The second group of results concerns financing and viability. We extend from a single return trajectory to a scenario set, and require trading rules to be jointly measurable and non-anticipative, with the same rule applied consistently across all scenarios. A market is called \emph{pathwise viable} when no nontrivial future withdrawal stream can be financed from arbitrarily small positive initial capital. We show that this is equivalent to the pointwise boundedness of attainable terminal wealth levels on every scenario. This is the natural probability-free counterpart of boundedness-in-probability in the classical theory. It is stronger, in the sense that no exceptional set of scenarios can be discarded, and expresses a genuinely robust requirement on the collection of possible paths.

These two groups of equivalences form the two parts of our main result, Theorem~\ref{thm:pathwise-finite-resolution-cornerstone}. The first is the \emph{growth--num\'eraire layer}; the second is the \emph{viability--boundedness layer}. A central theme of the paper is that, in the general pathwise setting, these layers cannot be merged. In particular, the validity of the growth-num\'eraire layer does not by itself imply pathwise viability.

The reason these two layers can be merged in the classical setting lies in properties specific to continuous semimartingales. There, the residual in the canonical decomposition is a continuous local martingale. Once the structural condition for the num\'eraire is satisfied, every relative wealth process becomes the stochastic exponential of a local martingale, thus a nonnegative local martingale and a supermartingale. The reciprocal of a num\'eraire is a local martingale deflator; and the resulting supermartingale inequalities provide budget restrictions for financeable liabilities, and control attainable wealth in probability. This probabilistic mechanism is precisely the bridge joining local growth structure to global market viability in Theorem~2.31 of \cite{KaratzasKardaras2021}.

There is also a more elementary but equally important distinction. A continuous local martingale with zero quadratic variation is constant. Consequently, in the semimartingale setting, the residual cannot conceal a nontrivial finite variation direction that is invisible to covariation but exploitable by trading. No analogous statement holds for a general pathwise residual, which may well contain a smooth, nonconstant component with zero F\"ollmer quadratic variation. Such movement is invisible to the covariation-based growth characteristic, yet can generate arbitrarily large gains for suitable pathwise portfolios. Thus, the calculation of characteristics can be perfectly well-behaved, while attainable wealth levels remain unbounded.

Examples~\ref{ex:lfg-without-pointwise-boundedness} and~\ref{ex:multi-scenario-non-anticipative-separation} exhibit this phenomenon in complementary ways. The first uses a smooth return path and a finite-resolution Faber--Schauder trend extractor: the extracted trend and quadratic covariation vanish, yet portfolios constructed from the local slope of the residual generate arbitrarily large terminal wealth. The second removes the informational degeneracy of a singleton scenario set. An uncountable family of paths agrees up to a common branching time, while one sequence of jointly Borel-measurable, non-anticipative portfolio rules generates terminal wealth that diverges uniformly over all scenarios. Thus, the separation of the two layers is neither a technical artifact nor a consequence of advance knowledge along a single path; it records the loss of the local-martingale structure which, in semimartingale theory, prevents such hidden directional gains.

The paper develops a pathwise counterpart of the classical cornerstone theorem, with two parts. The {\it affirmative part} identifies the relations among finite growth opportunities, characteristic-based optimality, and driftless num\'eraire, and separately identifies viability with boundedness of attainable wealth. The {\it negative part} shows that these relations cannot in general be assembled into the single equivalence class familiar from the semimartingale theory. In this way, the pathwise formulation reveals not only what survives outside a probabilistic framework, but also the precise work performed by the local martingale structure in the classical result.

\smallskip
 
\noindent \emph{Preview}: The next section develops pathwise quadratic covariation, non-anticipative trend extraction, and finite-resolution characteristics, together with Faber--Schauder examples. Section~\ref{sec:portfolio-characteristics} introduces portfolio characteristics, wealth-generating F\"ollmer portfolios, growth optimality, and driftless num\'eraires. Section~\ref{sec:scenario-wise-viability} develops scenario-wise financing and pathwise viability, proves the equivalence with pointwise boundedness of terminal wealth, and concludes with the pathwise cornerstone theorem and two examples separating its two layers.

\section{Pathwise characteristics}

We start by developing the F\"ollmer theory of quadratic covariations along a given sequence of refining partitions of a fixed time-interval. We then use this theory to generate non-anticipative decompositions of continuous paths at finite resolution; these mimic the familiar semimartingale decompositions of stochastic analysis, but are completely devoid of probabilistic considerations.

\subsection{F\"ollmer's quadratic covariation along refining partitions}  \label{subsec:follmer-quadratic-covariation}

We recall from \cite{follmer1981} the notion of ``pathwise quadratic covariation'' along refining sequences of partitions of a given time-interval $[0,T]$. Throughout this subsection, $T>0$ is a fixed real number; and continuous functions $x:[0, T] \to \mathbb{R}^d$ are called ``paths''.

A \emph{refining sequence of partitions} of the interval $[0,T]$, is a sequence $\Pi=(\Pi_n)_{n\in\mathbb N}$, with $\Pi_n \subset \Pi_{n+1}$, $\forall \, n \in \N$ and
\begin{equation}    \label{def:refining-partitions}
    \Pi_n = \{t^{(n)}_k\}_{k=0}^{M_n+1}, \quad 0=t^{(n)}_0<t^{(n)}_1<\cdots<t^{(n)}_{M_n+1}=T, \quad |\Pi_n|:=\max_{1\le k\le M_n+1}(t^{(n)}_k-t^{(n)}_{k-1})\longrightarrow  0 
\end{equation}
as $n\to\infty$. In what follows, such a refining sequence $\Pi$ will be fixed, and the quantities introduced will be indexed by $\Pi$ in order to stress their dependence on this sequence.

\begin{definition}[Scalar Quadratic Variation]\label{def:scalar-QV}
    For each given continuous path $x\in C([0,T];\R)$ and $n\in\N$, we define the discrete quadratic variation measure
    \[
       {\bm \mu}_x^{(n)} := \sum_{k=1}^{M_n+1} \Big( x(t^{(n)}_{k})-x(t^{(n)}_{k-1}) \Big)^2 {\bm \delta}_{t^{(n)}_{k-1}}.
    \]
    We say that the path $x$ \emph{has finite quadratic variation} along $\Pi$, if the sequence $({\bm \mu}_x^{(n)})_{n \in \N}$ converges vaguely to a finite, non-atomic measure ${\bm \mu}_x$ on the Borel sets of $[0,T]$. \footnote{~That is, $\lim_{n \to \infty} \int_{[0,T]} f \, \d {\bm \mu}_x^{(n)} = \int_{[0,T]} f \, \d {\bm \mu}_x\,$ for every continuous $f: [0,T] \to \R.$} Then, the continuous, nondecreasing function
    \[
        [x]_\Pi(t):={\bm \mu}_x([0,t]), \qquad 0 \le t \le T.
    \]
    is the so-called ``F\"ollmer quadratic variation'' of the path $x$ over each interval $[0, t]$, along $\Pi$.
\end{definition}

We also note the interpretation
\[
    [x]_\Pi(t) = \lim_{n\to\infty} \sum_{\substack{k=1 \\ t^{(n)}_{k-1} \le t}}^{M_n+1}\Big( x(t^{(n)}_k \wedge t)- x(t^{(n)}_{k-1}) \Big)^2, \qquad 0 \le t \le T.
\]

\begin{definition}[Vector Quadratic Covariation] \label{def:vector-QC}
    Consider now a vector \(x=(x^1,\ldots,x^d)\in C([0,T];\mathbb R^d)\) of $d$ continuous functions on $[0, T]$. For each $n \in \N$ we define the discrete, matrix-valued covariation measure
    \[
       {\bm  \mu}^{(n)}_x := \sum_{k=1}^{M_n+1} \Big( x(t^{(n)}_k)- x(t^{(n)}_{k-1}) \Big) \Big( x(t^{(n)}_k)- x(t^{(n)}_{k-1}) \Big)^\top {\bm \delta}_{t^{(n)}_{k-1}}
    \]
 and say that the path \(x\) \emph{has finite quadratic covariation} along $\Pi$ if the entries of the sequence \(({\bm \mu}^{(n)}_x)_{n\in\N}\) converge vaguely to finite signed measures \({\bm \mu}^{ij}_x\), \(1\le i,j\le d\), and the limiting matrix-valued measure $\,{\bm \mu}_x=({\bm \mu}^{ij}_x)_{1\le i,j\le d}\,$ is nonnegative-definite (in the sense that ${\bm \mu}_x(B)$ is a symmetric nonnegative-definite matrix for every Borel subset $B$ of $[0,T]$) and continuous (in the sense that its diagonal measures ${\bm \mu}^{i,i}_x$, $i=1, \ldots, d$ are non-atomic).
\end{definition}

In this case, we consider the matrix of so-called ``F\"ollmer quadratic covariations''
\begin{equation}    \label{def:quadratic-covariation}
    [x,x]_\Pi(t) := \big([x^i,x^j]_\Pi(t)\big)_{1\le i,j\le d}, \qquad [x^i,x^j]_\Pi(t):={\bm \mu}^{ij}_x([0,t]); \qquad 0 \le t \le T,
\end{equation}
and note that the resulting matrix-valued path $[x,x]_\Pi$ is continuous, symmetric, with nonnegative-definite increments and
\[
    [\xi^\top x]_\Pi(t) = \xi^\top[x,x]_\Pi(t) \, \xi, \qquad \forall\,~~\xi\in\mathbb R^d, \qquad 0 \le t \le T.
\]
We note also that cross-variations in \eqref{def:quadratic-covariation} may be recovered by the polarization
\[
    [x^i,x^j]_\Pi = \frac12\Big([x^i+x^j]_\Pi-[x^i]_\Pi-[x^j]_\Pi\Big), \qquad 1 \le i, j \le d.
\]

\begin{definition}[The Class $Q_\Pi$ of Paths]  \label{def:Q-Pi}
    We denote by $Q_\Pi([0,T];\mathbb R^d)$ the collection of paths $x\in C([0,T];\mathbb R^d)$ with finite quadratic covariation along the sequence $\Pi$ of partitions. When $d=1$, we write simply $Q_\Pi([0,T])$.
\end{definition}

For a given path $x\in Q_\Pi([0,T];\mathbb R^d)$, each entry $[x^i, x^j]_\Pi$ of the matrix in \eqref{def:quadratic-covariation} is a continuous path of finite (first) variation. And because the matrix-valued measure ${\bm \mu}_x = \big({\bm \mu}_x^{i,j}\big)_{1\le i,j \le d}$ is nonnegative-definite, the total variations of the off-diagonal measures are dominated by the diagonal ones, in the sense
\begin{equation} \label{ineq:off-diagonal-domination}
    2\,\big| {\bm \mu}_x^{i,j} \big| \le \,{\bm \mu}_x^{i,i}+{\bm \mu}_x^{j,j}, \qquad 1 \le i,j \le d.
\end{equation}
This has a simple consequence, which we record in Lemma~\ref{lem:BV-zero-QV} below: ``adding a path of finite first variation does not change the quadratic variation''.

\begin{lemma}[Paths of Finite First Variation have Zero Quadratic Variation] \label{lem:BV-zero-QV}
    Every continuous path $y:[0,T]\to\mathbb R^d$ with finite first variation satisfies
    \[
        y \in Q_\Pi([0,T];\R^d), \qquad [y,y]_\Pi \equiv 0;
    \]
    and for every other path $x \in Q_\Pi([0,T];\mathbb R^d)$, we have
    \[
        [x,y]_\Pi\equiv0, \qquad [x+y,x+y]_\Pi=[x-y,x-y]_\Pi=[x,x]_\Pi.
    \]
\end{lemma}

\begin{proof}
 We focus on the scalar case $d=1$;   the statement for $d \ge 2$ can be deduced from it    via polarization. Then, since the path $y:[0, T]\to \R$ has finite first variation, we obtain
    \[
        \sum_{k=1}^{M_n+1} \big(y(t^{(n)}_k) - y(t^{(n)}_{k-1})\big)^2 \le \Big(\max_{1\le k \le M_n+1}\big|y(t^{(n)}_k) - y(t^{(n)}_{k-1})\big| \Big)\, \breve{y}(T) \longrightarrow 0
    \]
    as $n \to \infty$, which gives $[y]_\Pi \equiv 0$. We have denoted here by $\breve{y}(T) := |y|_{\mathrm{TV};[0, T]}$ the total variation of the continuous path $y$ on $[0, T]$. Likewise, for any path $x:[0,T]\to\R$ in $Q_\Pi([0,T])$, the Cauchy--Schwarz inequality gives
    \begin{align*}
        \sum_{k=1}^{M_n+1}\Big|\big(x(t^{(n)}_k) &- x(t^{(n)}_{k-1})\big) \big(y(t^{(n)}_k) - y(t^{(n)}_{k-1})\big) \Big| 
        \\
        &\le \sqrt{ \sum_{k=1}^{M_n+1}\big(x(t^{(n)}_k) - x(t^{(n)}_{k-1})\big)^2 \sum_{k=1}^{M_n+1}\big(y(t^{(n)}_k) - y(t^{(n)}_{k-1})\big)^2} \longrightarrow 0
    \end{align*}
    as $n\to\infty$,   leading to $[x,y]_\Pi \equiv 0$.
\end{proof}

We can state now the following fundamental result. 

\begin{lemma}[The F\"ollmer Pathwise It\^o Rule \cite{follmer1981}]    \label{lem:follmer-ito}
    Consider a path $x\in Q_\Pi([0,T];\mathbb R^d)$ and a function $f: \R^d \to \R$ of class $C^2$. Then the pathwise It\^o integral
    \begin{equation}    \label{eq:Ito-rule}
        \int_0^t \nabla f\big(x(s)\big)\,\d^\Pi x(s), \qquad 0 \le t \le T
    \end{equation}
    is well-defined as the limit 
    \[
        \lim_{n\to\infty} \sum_{\substack{k=1 \\ t^{(n)}_{k-1} \le t}}^{M_n+1} \nabla f\big(x(t^{(n)}_{k-1})\big) \cdot \big(x(t^{(n)}_k \wedge t)-x(t^{(n)}_{k-1})\big)
    \]
of non-anticipative Riemann sums, and satisfies the change-of-variable formula
    \begin{equation*}
        f\big(x(t)\big)-f\big(x(0)\big) = \int_0^t\nabla f\big(x(s)\big)\,\d^\Pi x(s) + \frac12 \sum_{i=1}^d \sum_{j=1}^d \int_0^t D_{ij}^2 f\big(x(s)\big)\,\d [x^i,x^j]_\Pi(s), \quad 0 \le t \le T.
    \end{equation*}
\end{lemma}

\begin{remark}  \label{rem:Ito-integrands}
    It is instructive to point out here, just how special this construction is. If $X$ is standard, scalar Brownian Motion on some filtered probability space $(\Omega, \mathcal F, \mathbb{P})$, $\mathbb F = (\mathcal F(t))_{0\le t \le T}$, the It\^o integral
    \[
        \int_0^t Y(s) \, \d X(s), \qquad 0 \le t \le T
    \]
    can be defined for every $\mathbb F$-progressively measurable $Y : [0,T]\times\Omega \to \R$ with $\int_0^T Y^2(s, \omega) \d s < \infty$, $\mathbb P$-a.e. $\omega \in \Omega$. In particular, $Y(s)$ can depend on the entire path $(X(u), \, 0 \le u \le s)$. 
    
    By contrast, the recipe in \eqref{eq:Ito-rule} works only for the very special integrands of the form $\,y(s) = \nabla f(x(s)), \, 0 \le s \le T\,,$ with $f$ of class $C^2$.
\end{remark}

The following extension of the It\^o--F\"ollmer formula is
\cite[Theorem~9]{schied2014}.

\begin{lemma}[It\^o--F\"ollmer Formula with Finite Variation Controls]  \label{lem:follmer-ito-fv-control}
    Let $x\in Q_\Pi([0,T];\mathbb R^d)$, let the components of $b=(b^1,\ldots,b^m)\in C([0,T];\mathbb R^m)$ be functions of finite first variation, and let $f\in C^{1,2}(\mathbb R^m\times\mathbb R^d)$. Then
    \[
        \int_0^t \nabla_x f\big(b(s),x(s)\big)^\top\,\d^\Pi x(s)
    \]
    exists as the uniform limit on $[0,T]$ of the corresponding left Riemann sums, and we have the change-of-variable formula
    \begin{align*}   
        f\big(b(t),x(t)\big)-f\big(b(0),x(0)\big) &= \sum_{\ell=1}^m \int_0^t \partial_{b_\ell} f\big(b(s),x(s)\big)\,\d b^\ell(s)
        +\int_0^t \nabla_x f\big(b(s),x(s)\big)^\top\,\d^\Pi x(s) 
        \\
        & \qquad +\frac12\sum_{i,j=1}^d \int_0^t \partial^2_{x_ix_j}f\big(b(s),x(s)\big) \,\d[x^i,x^j]_\Pi(s).
    \end{align*}
\end{lemma}

\medskip
 
\subsection{Non-anticipative decompositions with finite-resolution}   \label{subsec:nonanticipative-decomposition}

In the standard setting of Mathematical Finance, the cumulative returns $R=(R^1,\ldots,R^d)$ of a fixed collection consisting of $d$ financial assets, admit a canonical decomposition
\begin{equation}    \label{eq:R-canonical-decomposition}
    R=A+M,
\end{equation}
where $A = (A_1, \ldots, A_d)$ is a vector of continuous processes with finite first variation on compact intervals, and $M=(M_1,\ldots,M_d)$ is a vector of continuous local martingales. The former are understood as the ``drift'', or ``trend'', components of the cumulative returns, and the latter as the ``noise'' components. 

In a purely pathwise setting, there is no such probabilistic picture. A continuous path $R$ does \emph{not} come equipped with a drift/noise decomposition of the form \eqref{eq:R-canonical-decomposition}. For this reason, we shall replace here the drift component by a \emph{non-anticipative trend characteristic}  extracted from past observations of the trajectory of returns $R$, and provide a pathwise analogue of the semimartingale decomposition \eqref{eq:R-canonical-decomposition}.

This will be done as follows: we fix an ``observation window'' length $h>0$, a ``trading horizon'' $T \gg h$, as well as a refining sequence $\Pi=(\Pi_n)_{n\in \mathbb N}$ of partitions of the interval $[0,T]$; and consider a continuous, $\R^d$-valued function $R$ defined on $[-h,T]$ and of finite quadratic covariation along $\Pi$ on $[0, T]$:
\begin{equation}    \label{def:return}
    R \in C([-h,T];\R^d), \qquad R|_{[0,T]} \in Q_\Pi([0,T];\R^d).
\end{equation}
We think of $R$ as the path, or trajectory, of $d$ financial assets' returns, which are observed over the time interval $[-h,T]$: we interpret $[-h,0)$ as a ``pre-trading observation period'', and $[0, T]$ as the actual trading period. And for each time $t\in[0,T]$ in the trading period, we consider the ``rescaled past window path'' $R^{t,h}  \in C([0,1];\mathbb R^d)$ given by
\begin{equation}    \label{rescaled-past-window-return-path}
    R^{t,h}(u):=R(t-h+hu),\qquad 0 \le u \le 1.
\end{equation}

\begin{definition}[Non-anticipative Trend Extractor]    \label{def:nonanticipative-trend-extractor}
    Given an $\R^d$-valued path $R$ as in \eqref{def:return}, a refining sequence $\Pi = (\Pi_n)_{n \in \N}$ of partitions of $[0, T]$, and an observation window $h>0$, a Borel-measurable map $\mathcal D:C([0,1];\mathbb R^d) \to \mathbb R^d$ is called \emph{non-anticipative trend extractor for} $(R,\Pi,h)$, if the so-called ``trend signal''
    \begin{equation}    \label{def:trend-signal}
        \alpha^{h,\mathcal D}(t) := \frac1h \,\mathcal D(R^{t,h}), \qquad 0 \le t \le T
    \end{equation}
    belongs to the space $\mathbb{L}^1([0,T];\mathbb R^d)$.
\end{definition}

The adjective ``non-anticipative'' reflects the fact that the signal in \eqref{def:trend-signal} is determined, for each $t \in [0, T]$, only on the basis of the ``past-window'' trajectory $(R(s), t-h \le s \le t)$, as is clear from \eqref{rescaled-past-window-return-path}. From the trend signal of \eqref{def:trend-signal} we introduce now the \emph{cumulative trend function}
\begin{equation}    \label{def:cumulative-trend-function}
    A^{h, \mathcal D}(t):=\int_0^t \alpha^{h,\mathcal D}(s)\, \d s, \qquad 0 \le t \le T
\end{equation}
as well as the \emph{residual function}
\begin{equation}    \label{def:residual-function}
    M^{h,\mathcal D}(t) := R(t)-R(0)-A^{h,\mathcal D}(t), \qquad 0 \le t \le T;
\end{equation}
and note that $A^{h,\mathcal D}(t)$ depends on the collection of past sliding windows, deployed up to time $t$,  hence on the restricted path $(R(u), -h \le u \le t)$; but {\it never} on values of $R(\cdot)$ after time $t$. The resulting decomposition
\begin{equation}    \label{eq:resulting-decomposition}
    R(t) = R(0) + A^{h,\mathcal D}(t)+M^{h,\mathcal D}(t),\qquad 0 \le t \le T
\end{equation}
of the continuous path $R $ at time $t$, consists of the component $A^{h, \mathcal D}(t)$ in \eqref{def:cumulative-trend-function} obtained by accumulating over $[0,t]$  the trend signal of \eqref{def:trend-signal} generated by the extractor $\mathcal D$, and of the residual $M^{h,\mathcal D}(t)$ in \eqref{def:residual-function}. Because $\alpha^{h, \mathcal D} $ belongs to $\,\mathbb{L}^1([0,T];\mathbb R^d)$, the function $A^{h,\mathcal D}$ is absolutely continuous relative to Lebesgue measure, therefore of finite first variation on the compact interval $[0,T]$. Applying Lemma~\ref{lem:BV-zero-QV} to the continuous, finite-variation path $R(0)+A^{h,\mathcal D}   =: B $, we deduce that the continuous path $M^{h,\mathcal D}  = R -B $ in \eqref{def:residual-function} has finite quadratic covariation along the sequence of partitions $\Pi$, i.e.,
\[
    M^{h, \mathcal D}   \in Q_\Pi([0,T];\R^d),
\]
and that
\begin{equation}    \label{eq:M-R-quadratic-covariation}
    C^{h, \mathcal D}  := \big[M^{h,\mathcal D},M^{h,\mathcal D} \big]_\Pi   =\big[R,R\big]_\Pi  \,.
\end{equation}
In other words, the residual process $M^{h,\mathcal D} $ inherits the entire ``second-order roughness'', i.e., the F\"ollmer quadratic covariation, of the original return path $R$.

The components of the resulting pair $(A^{h,\mathcal D},C^{h,\mathcal D})$ of \eqref{def:cumulative-trend-function}, \eqref{eq:M-R-quadratic-covariation} play now the role of \emph{finite-resolution pathwise characteristics}: the function $A^{h, \mathcal D} $ is the finite-variation trend characteristic induced by the chosen extractor $\mathcal D$, while $C^{h,\mathcal D} $ is the F\"ollmer quadratic covariation as in  \eqref{def:quadratic-covariation} carried by both the residual $M^{h,\mathcal D} $ and by the original return path $R$. 

\medskip
\noindent
\emph{Operational Clock}: We introduce now an ``operational clock'', very much in accordance with the role played by this notion in the continuous semimartingale setting of \cite{KaratzasKardaras2021}. Recalling the notation of \eqref{def:cumulative-trend-function}, \eqref{eq:M-R-quadratic-covariation}, we define the internal clock function
\begin{equation}    \label{def:O-clock}
    O^{h, \mathcal D}(t) := \sum_{i=1}^d \Big( \breve{A}^{h,\mathcal D}_i(t) + C^{h,\mathcal D}_{ii}(t) \Big), \qquad 0 \le t \le T,
\end{equation}
with $\breve{A}(t)$ denoting the total variation on $[0,t]$ of the continuous function $A:[0, T] \to \R$. The absolute continuity of $A^{h,\mathcal D}(\cdot)$, and the fact that the continuous, matrix-valued function $C^{h,\mathcal D}(\cdot)$ has nonnegative-definite increments, imply that $O^{h, \mathcal D}$ in \eqref{def:O-clock} is continuous and nondecreasing. In addition, the signed measure induced by each $A^{h, \mathcal D}_i  $ in \eqref{def:cumulative-trend-function} is absolutely continuous with respect to the one induced by $O^{h,\mathcal D} $; whereas, by the domination property \eqref{ineq:off-diagonal-domination}, the same holds for the signed covariation measure induced by each $C^{h, \mathcal D}_{ij}$.

\smallskip
We   define now one-sided Radon-Nikod\'ym density to introduce Borel-measurable ``rate'' functions in \eqref{def:non-anticipative-local-characteristics} below, with respect to the measure induced by $O^{h, \mathcal D}$, in a non-anticipative manner.

\begin{definition}[Non-anticipative Density]
\label{def:non-anticipative-one-sided-density}
    Let $E$ be a finite-dimensional Euclidean space, let $F:[0,T] \to E$ be continuous and of finite variation, and let $O:[0,T] \to [0,\infty)$ be continuous and nondecreasing.  For $r\in\mathbb Q\cap(0,\infty)$, set
    \[
        q_r^{F\mid O}(t) :=
        \begin{cases}
            \displaystyle
            \frac{F(t)-F((t-r)\vee0)}{O(t)-O((t-r)\vee0)},
            & \qquad O(t)>O((t-r)\vee0),\\[1.2em]
            \qquad 0, & \qquad O(t)=O((t-r)\vee0),
        \end{cases}
    \]
    and define the \emph{non-anticipative $O$-density} of $F$ as
    \begin{equation} \label{def:non-anticipative-density-operator}
        \mathrm D_O^-F(t) :=
        \begin{cases}
            \displaystyle
            \lim_{\substack{r\downarrow0\\ r\in\mathbb Q}}
            q_r^{F\mid O}(t),
            & \qquad \text{if this finite limit exists in $E$},\\[0.8em]
            \qquad 0, & \qquad \text{otherwise}.
        \end{cases}
    \end{equation}
\end{definition}

\begin{lemma}[Non-anticipative Radon--Nikod\'ym Representative]   \label{lem:non-anticipative-radon-nikodym-representative}
    Let $F$ and $O$ be as in Definition~\ref{def:non-anticipative-one-sided-density}, and suppose that the $E$-valued measure induced by $F$ is absolutely continuous with respect to the measure induced by $O$. Then $\mathrm D_O^-F$ of \eqref{def:non-anticipative-density-operator} is Borel-measurable, and
    \begin{equation}    \label{eq:non-anticipative-rn-representation}
        F(t)-F(0)=\int_0^t \mathrm D_O^-F(s)\,\d O(s), \qquad 0\le t\le T.
    \end{equation}
    Equivalently, $\mathrm D_O^-F$ is a version of $\d F/\d O$.

    If $E=\mathbb R^{d\times d}$, $F(0)=0$, and $F$ has symmetric nonnegative-definite increments, then $\mathrm D_O^-F(t)$ is symmetric and nonnegative definite for every $t\in[0,T]$.
\end{lemma}

\begin{proof}
    Let $\mu_O$ be the finite Borel measure induced by $O$. By the one-sided Lebesgue differentiation theorem for finite Borel measures on $\mathbb R$, for every scalar component of $F$ the quotient in \eqref{def:non-anticipative-density-operator} converges to the corresponding Radon--Nikod\'ym derivative for $\mu_O$-almost every $t$.  Applying this componentwise gives \eqref{eq:non-anticipative-rn-representation}.

    For each rational $r>0$, the map $t\mapsto q_r^{F\mid O}(t)$ is Borel-measurable. The set on which the rational one-sided limit exists is a Borel set by the Cauchy criterion, and the limit on this set is Borel; the convention of assigning $0$ on its complement therefore yields a Borel-measurable map.

    In the matrix-valued case, whenever the denominator is positive, $q_r^{F\mid O}(t)$ is a positive scalar multiple of a symmetric nonnegative-definite increment of $F$.  The cone of symmetric nonnegative-definite matrices is closed, and the exceptional value is chosen to be the zero matrix.  The final assertion follows.
\end{proof}

We now use Lemma~\ref{lem:non-anticipative-radon-nikodym-representative} to select versions of the local characteristics defined pointwise. Namely, we set
\begin{equation}    \label{def:non-anticipative-local-characteristics}
    a^{h,\mathcal D}(t) :=\mathrm D_{O^{h,\mathcal D}}^-A^{h,\mathcal D}(t),
    \qquad
    c^{h,\mathcal D}(t) :=\mathrm D_{O^{h,\mathcal D}}^-C^{h,\mathcal D}(t).
\end{equation}
Then $a^{h,\mathcal D}$ and $c^{h,\mathcal D}$ are Borel measurable, $c^{h,\mathcal D}(t)$ is symmetric and nonnegative definite for every $t$, and
\begin{equation*}
    A^{h,\mathcal D}(t) = \int_0^t a^{h,\mathcal D}(s)\,\d O^{h,\mathcal D}(s),
    \qquad
    C^{h,\mathcal D}(t) = \int_0^t c^{h,\mathcal D}(s)\,\d O^{h,\mathcal D}(s).
\end{equation*}
Moreover, the definition of the clock gives
\[
    |a_i^{h,\mathcal D}(t)| \le 1, \qquad c^{h,\mathcal D}(t)\succeq0, \qquad \operatorname{Tr}\big(c^{h,\mathcal D}(t)\big) \le 1
\]
for every $t \in [0,T]$. Indeed, the same inequalities hold for every difference quotient in
\eqref{def:non-anticipative-density-operator}.

\begin{definition}[Non-anticipatively Selected Pathwise Local Characteristics]  \label{def:pathwise-local-characteristics-extractor}
    For a return path $R$, a partition sequence $\Pi$, a window length $h>0$, and a non-anticipative trend extractor $\mathcal D$, the pair $(a^{h,\mathcal D},c^{h,\mathcal D})$ in \eqref{def:non-anticipative-local-characteristics} is called the pair of \emph{non-anticipatively selected pathwise local characteristics} induced by $(\mathcal D,h,\Pi)$.
\end{definition}

\medskip

\subsection{Examples of non-anticipative trend extractors}  \label{subsec:FS-examples-nonanticipative-extractors}

We introduce now canonical trend extractors via the familiar Faber--Schauder representation of continuous paths.

Let us denote by $\Pi^*=(\Pi^*_n)_{n\in \mathbb{N}}$ the sequence of dyadic partitions of the interval $[0,1]$, i.e., with $t^{(n)}_k=k2^{-n}$, $k = 0, 1, \ldots, 2^n$ and $T=1$ in \eqref{def:refining-partitions}, and consider the dyadic Haar functions
\begin{equation*}
    H^{(n)}_{k}(u) := 2^{n/2} H(u2^n-k), \qquad n \in \N_0, \quad k\in I_n:=\{0,1,\dots,2^n-1\}
\end{equation*}
for
\[
    H(u):=
    \begin{cases}
    ~1, & u\in [0,\frac12)\\
    -1, & u\in [\frac12,1)\\
    ~0, & \text{otherwise}
    \end{cases}.
\]
With $\N_{-1}:=\{-1,0,1,\dots\}$, $I_{-1}:=\{0\}$ and $H^{(-1)}_{0}(t) \equiv 1 $ for $0 \le t \le 1$, the system $\{H^{(n)}_{k}(\cdot) : k\in I_n\}_{n \in \N_{-1}}$ is the \emph{Haar Orthonormal Basis} of $L^2([0,1];\mathbb R)$. In terms of this system, the associated Faber--Schauder functions are  
\begin{equation*}
    S^{(n)}_{k}(u):=\int_0^u H^{(n)}_{k}(s)\,\d s, \quad 0 \le u \le 1 \qquad \text{for } n\in\N_{-1},\ k\in I_n.
\end{equation*}

For any given   $x\in C([0,1];\mathbb R^d)$, we have then the corresponding Faber--Schauder expansion
\begin{equation}    \label{eq:FS-expansion}
    x(u)=x(0)+u\big(x(1)-x(0)\big) + \sum_{m \in \N_0} \sum_{k=0}^{2^m-1} \vartheta^{(x)}_{m,k} S^{(m)}_{k}(u), \qquad 0 \le u \le 1,
\end{equation}
where $\vartheta^{(x)}_{m,k} = \big(\vartheta^{(x,i)}_{m,k}\big)_{i=1, \ldots, d}$ is the vector of Faber--Schauder coefficients
\begin{equation*}
    \vartheta^{(x,i)}_{m,k} \,:= \, 2^{m/2} \bigg(2x_i\Big(\frac{2k+1}{2^{m+1}}\Big)-x_i\Big(\frac{k}{2^m}\Big)-x_i\Big(\frac{k+1}{2^m}\Big)\bigg),\qquad i=1,\ldots,d.
\end{equation*}
We write \eqref{eq:FS-expansion} as $x = P_Nx + \Phi_N x$, where
\begin{equation}    \label{def:PNx}
    (P_Nx)(u) := x(0) + u\big(x(1)-x(0)\big) + \sum_{m=0}^{N-1} \sum_{k=0}^{2^m-1} \vartheta^{(x)}_{m,k} S^{(m)}_{k}(u), \qquad 0 \le u \le 1
\end{equation}
is the \emph{$N$-Resolution Faber--Schauder projection} of the path $x$, and 
\[
    (\Phi_Nx)(u) := x(u)-(P_Nx)(u) = \sum_{m \ge N} \sum_{k=0}^{2^m-1} \vartheta^{(x)}_{m,k} S^{(m)}_{k}(u), \qquad 0 \le u \le 1
\]
the \emph{residual projection}. This contains no affine component and no mode of oscillation level coarser than $N$, so we interpret it as the fine-scale fluctuation, or ``noise-like component'', that remains after removing the finite-resolution trend $P_Nx$. We stress that this interpretation is purely pathwise, and involves no probabilistic considerations.

We apply this decomposition to the rescaled past-window return path $R^{t,h} = \big(R^{t,h}(u)\big)_{0 \le u \le 1}$ of \eqref{rescaled-past-window-return-path}, for a given trading time $t \in [0, T]$. In the resulting decomposition
\[
    R^{t,h}=P_NR^{t,h}+\Phi_NR^{t,h},
\]
the term $P_NR^{t,h}$ is piecewise-affine, and so has zero quadratic variation along dyadic partitions of the interval $[0, 1]$. Thus, whenever $R^{t,h}$ admits dyadic quadratic covariation, we have 
\[
    \big[\Phi_NR^{t,h}, \Phi_NR^{t,h} \big]_{\Pi^*} = \big[R^{t,h}, R^{t,h}\big]_{\Pi^*}, \qquad N \in \mathbb N.
\]

\noindent \emph{Remark}: This identity is \emph{local} in the window variable; the global covariation is given as in \eqref{eq:M-R-quadratic-covariation}, namely,
\[
    C^{h,\mathcal D} = \big[M^{h,\mathcal D}, M^{h,\mathcal D}\big]_\Pi = \big[R,R\big]_\Pi.
\]

\bigskip

Here are two examples of non-anticipative trend extractors generated by these 
expansions.

\begin{example}[Last-Cell Faber--Schauder Slope]    \label{ex:last-FS-slope}
    In the notation of \eqref{def:PNx}, the Last-Cell Faber--Schauder slope extractor at resolution $N$, is given by
    \begin{equation}    \label{eq:last-cell-extractor}
        \mathcal D_N^{\mathrm{last}}(x) := 2^N \bigg( \big(P_Nx \big)(1)-\big(P_Nx\big) \Big(1-\frac{1}{2^{N}}\Big) \bigg), \qquad  x\in C([0,1]; \mathbb R^d).
    \end{equation}
    Its associated calendar-time trend signal, in the manner of \eqref{def:trend-signal}, is then
    \begin{equation}    \label{eq:last-cell-signal}
        \alpha^{h, N}_{\mathrm{last}}(t) =  \frac1h \mathcal D_N^{\mathrm{last}}(R^{t,h}) = \frac{2^N}{h} \bigg( \big(P_NR^{t,h}\big)(1)-\big(P_NR^{t,h}\big)\Big(1-\frac{1}{2^N}\Big) \bigg), \qquad 0 \le t \le T.
    \end{equation}
    This non-anticipative trend extractor $\mathcal D_N^{\mathrm{last}}$ uses the slope of the Faber--Schauder reconstruction \emph{over the most recent dyadic cell} in the past-observation window. Since $P_Nx$ in \eqref{def:PNx} depends only on finitely-many values of the underlying path, the mapping $x \mapsto \mathcal D_N^{\mathrm{last}}(x)$ of \eqref{eq:last-cell-extractor}  is continuous under the topology of uniform convergence; as a consequence, for every continuous return path $R$, the signal $t\mapsto \alpha^{h,N}_{\mathrm{last}}(t)$ in \eqref{eq:last-cell-signal} is continuous, thus in $\mathbb{L}^1([0,T];\mathbb R^d)$.
    
    Consequently, \eqref{eq:last-cell-extractor} defines a non-anticipative trend extractor $\mathcal D_N^{\mathrm{last}}$ for $(R,\Pi,h)$ at resolution $N$, whenever $R|_{[0,T]}\in Q_\Pi([0,T];\mathbb R^d)$, as in Definition~\ref{def:nonanticipative-trend-extractor}.
\end{example}

\begin{example}[Weighted Faber--Schauder Slope] \label{ex:weighted-FS-slope}
    Consider now a ``weight function'' $w:  [0,1] \to [0, \infty)$ satisfying $\int_0^1 w(u)\,\d u=1$, and define
    \begin{equation}    \label{def:weighted-FS-slope}
        \mathcal D_N^w(x) := \int_0^1 w(u) \frac{\partial}{\partial u} \big(P_N x \big)(u)\,\d u, \qquad x \in C([0,1];\R^d).
    \end{equation}
    Here, the indicated derivative exists for Lebesgue-a.e. $u \in [0, 1]$, because the function $P_Nx$ in \eqref{def:PNx} is piecewise-affine. The associated calendar-time trend signal is then defined as 
    \begin{equation}    \label{def:weighted-FS-slope-trend}
        \alpha^{h,N}_w(t) = \frac1h \, \mathcal D_N^w(R^{t,h}), \qquad 0 \le t \le T.
    \end{equation}
    When $w(\cdot)$ is supported near $u=1$, this extractor emphasizes recent-past information; and the expressions \eqref{eq:last-cell-extractor}, \eqref{eq:last-cell-signal} correspond to the special choice 
    \[
        w(u)=2^N\mathbf 1_{(1-2^{-N},1]}(u)\,.
    \]

    Because $P_Nx$ in \eqref{def:PNx} depends only on finitely-many values of the underlying path, and $u \mapsto \frac{\partial}{\partial u}\big(P_Nx\big)(u)$ is a linear combination of finitely-many Haar functions, the recipe \eqref{def:weighted-FS-slope} defines a continuous linear functional of the Faber--Schauder coefficients up to level $N-1$. As a consequence, given any  continuous path $R$, the signal in \eqref{def:weighted-FS-slope-trend} is continuous and in $\mathbb{L}^1([0,T]; \mathbb R^d)$. Therefore, $\mathcal D_N^w$ is a non-anticipative trend extractor for $(R,\Pi,h)$ at resolution $N$, whenever $R|_{[0,T]}\in Q_\Pi([0,T];\mathbb R^d)$.
\end{example}

\medskip

\section{Portfolios and their characteristics}   \label{sec:portfolio-characteristics}

We are now in a position to introduce and study \emph{portfolios} in a financial market, whose cumulative asset returns are of the form \eqref{eq:resulting-decomposition} introduced and elaborated upon in the previous section. In particular, we study the \emph{return}, \emph{quadratic (co)variation}, and \emph{growth} characteristics of such portfolios, and introduce notions such as ``growth optimality'' and the ``num\'eraire property''. The development follows closely, and parallels, Chapters 1 and 2 of \cite{KaratzasKardaras2021}; and ``translates'' the probabilistic framework of that monograph to the pathwise setting adopted here.

\medskip

\subsection{Portfolios and Growth}  \label{subsec:portfolios-growth}

Throughout this subsection we fix an  ``observation  window'' of length $h>0$, a ``trading horizon" $[0,T]$ with $T   \gg   h$ as in subsection~\ref{subsec:nonanticipative-decomposition}, a return path $R\in C([-h,T];\mathbb R^d)$, a refining sequence $\Pi$ of partitions of $[0, T]$ as in \eqref{def:refining-partitions}, and a non-anticipative trend extractor $\mathcal D : C([0, 1];\R^d) \to \R^d$ for $(R,\Pi,h)$ as in Definition~\ref{def:nonanticipative-trend-extractor}. We recall also the corresponding finite-resolution decomposition
\[
    R(t) = R(0) + A^{h,\mathcal D}(t) + M^{h, \mathcal D}(t), \qquad 0 \le t \le T
\]
in the manner of \eqref{def:cumulative-trend-function}--\eqref{def:residual-function}; as well as the quadratic covariation $C^{h,\mathcal D} $ of \eqref{eq:M-R-quadratic-covariation}, the operational clock of \eqref{def:O-clock} associated with the pair $(A^{h,\mathcal D}, C^{h, \mathcal D})$, and the \emph{pathwise local characteristics} $(a^{h, \mathcal D}, c^{h,\mathcal D})$ from Definition~\ref{def:pathwise-local-characteristics-extractor}.

In terms of these characteristics, we define now \emph{portfolios} and their associated \emph{return}, \emph{covariation}, and \emph{growth} characteristics.

\begin{definition}[Portfolios]  \label{def:portfolios}
    A Borel-measurable function $\pi:[0,T]\to\mathbb R^d$ is called   \emph{portfolio relative to} $(\mathcal D,h,\Pi)$, if
    \[
        \int_0^T \Big(|\pi(t)^\top a^{h,\mathcal D}(t)|+\pi(t)^\top c^{h,\mathcal D}(t) \pi(t)\Big)\, \d O^{h,\mathcal D}(t)<\infty.
    \]
    We denote by $\mathfrak P^{h,\mathcal D}$ the collection of all portfolios.
\end{definition}

The terminology in this definition comes, of course, from the eventual interpretation of the components of the vector $\pi(t) = (\pi_1(t), \ldots, \pi_d(t))^{\top}$ as   ``proportions of current wealth'' invested at time $t$ in each of $d$ financial assets with local characteristics $(a^{h, \mathcal D}, c^{h, \mathcal D})$. We shall elaborate, and build, on this interpretation in Subsection~\ref{subsec:wealth-level}. In the meantime we develop, in the present subsection and the next, some purely formal but important consequences of this definition.

\smallskip
We note first, that the collection $\mathfrak P^{h, \mathcal D}$ is closed under finite linear combinations. Moreover, it is checked easily that, if $\pi\in\mathfrak P^{h,\mathcal D}$, and if $\rho:[0,T]\to\mathbb R^d$ is Borel-measurable and bounded, then $\pi+\rho \in\mathfrak P^{h, \mathcal D}$.

\begin{definition}[Cumulative Portfolio Returns and Covariations]    \label{def:portfolio-return-covariation}
    For portfolios $\pi,\rho\in\mathfrak P^{h, \mathcal D}$, we define the cumulative \emph{return} and cumulative \emph{covariation} characteristics
    \begin{equation}    \label{def:cumulative-return-covariation}
        A^{h, \mathcal D}_\pi(t) := \int_0^t\pi(s)^\top\, \d A^{h, \mathcal D}(s) = \int_0^t \pi(s)^\top a^{h,\mathcal D}(s)\,\d O^{h,\mathcal D}(s), \qquad 0 \le t \le T,
    \end{equation}
    \begin{equation}    \label{def:cumulative-covariation}
        C^{h, \mathcal D}_{\pi,\rho}(t) := \int_0^t \pi(s)^\top c^{h, \mathcal D}(s) \, \rho(s) \, \d O^{h,\mathcal D}(s), \qquad 0 \le t \le T
    \end{equation}
    respectively. We write $C^{h,\mathcal D}_\pi \equiv C^{h,\mathcal D}_{\pi,\pi}$ and note $A^{h, \mathcal D}_{e_i} \equiv A^{h,\mathcal D}_i$, $C^{h,\mathcal D}_{e_i,e_j} \equiv C^{h,\mathcal D}_{ij}$, where $e_i$ denotes the $i$-th unit vector in $\mathbb R^d$.
\end{definition}

It is important to note that the covariation characteristic $C^{h, \mathcal D}_{\pi, \rho}$ is well-defined as in \eqref{def:cumulative-covariation}; indeed,
\[
    \big|\pi(t)^\top c^{h,\mathcal D}(t) \rho(t)\big| \le \Big(\pi(t)^\top c^{h,\mathcal D}(t)\,\pi(t)\Big)^{1/2} \Big(\rho(t)^\top c^{h,\mathcal D}(t) \, \rho(t)\Big)^{1/2},
\]
because $c^{h,\mathcal D}(\cdot)$ is nonnegative definite $\d O^{h,\mathcal D}$-a.e., and the right-hand side is $\d O^{h,\mathcal D}$-integrable. In particular, $C^{h,\mathcal D}_{\pi, \rho}(\cdot)$ in \eqref{def:cumulative-covariation} is of finite first variation, while $C^{h,\mathcal D}_\pi(\cdot)$ is nondecreasing.

\begin{definition}[Portfolio Growth]    \label{def:portfolio-growth}
    For a given portfolio $\pi\in\mathfrak P^{h,\mathcal D}$, we define the associated \emph{growth rate}
    \begin{equation}    \label{def-gamma-pi}
        \gamma^{h,\mathcal D}_\pi(t) := \pi(t)^\top a^{h,\mathcal D}(t)-\frac12 \,\pi(t)^\top c^{h,\mathcal D}(t) \, \pi(t), \qquad 0 \le t \le T
    \end{equation}
    with respect to the clock $O^{h, \mathcal D}(\cdot)$, and the associated \emph{cumulative growth}
    \begin{equation}    \label{def-Gamma}
        \Gamma^{h, \mathcal D}_\pi(t) := \int_0^t \gamma^{h,\mathcal D}_\pi(s)\,\d O^{h, \mathcal D}(s) = A^{h,\mathcal D}_\pi(t)-\frac12 \,C^{h,\mathcal D}_\pi(t), \qquad 0 \le t \le T.
    \end{equation}
    And for any portfolio $\rho\in\mathfrak P^{h, \mathcal D}$, we define the \emph{relative cumulative growth}
    \begin{equation}    \label{def:relative-cumulative-growth}
        \Gamma^{h, \mathcal D}_{\pi | \rho}(t) := \Gamma^{h, \mathcal D}_\pi(t) - \Gamma^{h,\mathcal D}_\rho(t), \qquad 0 \le t \le T.
    \end{equation}
\end{definition}

The following representation for this quantity follows by direct computation.

\begin{lemma}[Relative Growth Identity]\label{lem:relative-growth-identity}
    For any two given portfolios $\pi,\rho$ in $\mathfrak P^{h, \mathcal D}$, the relative cumulative growth of \eqref{def:relative-cumulative-growth} is given as
    \[
        \Gamma^{h,\mathcal D}_{\pi | \rho}(t) = \int_0^t \big(\pi(s)-\rho(s)\big)^\top \big(a^{h,\mathcal D}(s) -c^{h,\mathcal D}(s) \rho(s)\big)\,\d O^{h,\mathcal D}(s)-\frac12 \, C^{h,\mathcal D}_{\pi-\rho}(t)
    \]
    for $0 \le t \le T$; or equivalently, at the level of $O^{h,\mathcal D}$-densities, as
    \[
        \gamma^{h, \mathcal D}_\pi(t) -\gamma^{h, \mathcal D}_\rho(t) = \big(\pi(t)-\rho(t)\big)^\top \bigg( a^{h, \mathcal D}(t) - c^{h,\mathcal D}(t)\rho(t) -\frac12 \,c^{h,\mathcal D}(t)\big(\pi(t)-\rho(t)\big)\bigg).
    \]
\end{lemma}

The expression \eqref{def-gamma-pi} for the instantaneous growth rate, suggests defining the \emph{maximal growth rate}, corresponding to a given pair $(h, \mathcal D)$, as the $[0,\infty]$-valued function
\begin{equation}    \label{def:maximal-growth-rate}
    g^{h,\mathcal D}(t) := \sup_{p\in\mathbb R^d} \Big(p^\top a^{h,\mathcal D}(t)-\frac12 \, p^\top c^{h,\mathcal D}(t) \,p \Big), \qquad 0 \le t \le T;
\end{equation}
and the \emph{cumulative maximal growth} as
\begin{equation}    \label{def:cumulative-maximal-growth}
    G^{h,\mathcal D}(t) := \int_0^t g^{h,\mathcal D}(s)\,\d O^{h,\mathcal D}(s), \qquad 0 \le t \le T.
\end{equation}

It is important to note here, that the supremum in \eqref{def:maximal-growth-rate} can be taken equivalently over $\mathbb{Q}^d$; which shows that the function $g^{h, \mathcal D}$ is Borel-measurable.

\begin{definition}[Locally Finite Growth]    \label{def:finite-resolution-locally-finite-growth}
    The return path $R \in C([-h, T];\R^d)$ is said to have \emph{locally finite growth} on $[0,T]$, relative to the triple $(\mathcal D,h,\Pi)$, if
    \[
        G^{h, \mathcal D}(T)<\infty.
    \]
\end{definition}

We digress now, in order to provide an elementary linear-algebraic characterization for the maximal growth rate $g^{h, \mathcal D}(\cdot)$ of \eqref{def:maximal-growth-rate}, in the spirit of \cite[Chapters 2.1.2, 2.1.3]{KaratzasKardaras2021}. For any symmetric, nonnegative definite $(d\times d)$-matrix $c$, we denote by $c^\dagger$ the symmetric,  nonnegative definite matrix that inverts $c$ on $\operatorname{Range}(c)$, and vanishes on $\operatorname{Ker}(c)$. To wit, if
\[
    c=U\operatorname{diag}(\lambda_1,\ldots,\lambda_d)\,U^\top
\]
is an orthogonal diagonalization of $c$ with $\lambda_i\ge0$; then we denote the pseudo-inverse of $c$ by
\[
    c^\dagger:=U\operatorname{diag}(\lambda_1^\dagger,\ldots,\lambda_d^\dagger) \, U^\top, \qquad \lambda_i^\dagger:= \begin{cases}
    1/\lambda_i, & \lambda_i>0,\\
    ~0, & \lambda_i=0,
\end{cases}
\]
In this manner, $cc^\dagger=c^\dagger c$ is the orthogonal projection onto $\operatorname{Range}(c)$; and $cc^\dagger a=a$ holds if, and only if, $a\in\operatorname{Range}(c)$. 

We have then the following elementary result.

\begin{lemma}[Pointwise Maximal Growth]\label{lem:pointwise-maximal-growth}
    With $a\in\mathbb R^d$ and symmetric, nonnegative-definite $(d\times d)$-matrix $c$, we have
    \[
        \sup_{p\in\mathbb R^d}\Big(p^\top a-\frac12 \,p^\top c\,p\Big)<\infty \qquad \iff \qquad a\in \operatorname{Range}(c).
    \]
    In that case, the above supremum equals $\frac12 a^{\top}c^{\dagger}a$, and $\{c^\dagger a+\zeta:\zeta\in \operatorname{Ker}(c)\}$ is the set of its maximizers; whereas, if $a\notin\operatorname{Range}(c)$, this supremum is infinite.
\end{lemma}

\begin{proof}
    Since $c$ is symmetric and nonnegative-definite, $\operatorname{Range}(c)^\perp=\operatorname{Ker}(c)$. If $a\notin\operatorname{Range}(c)$, there exists $\zeta\in \operatorname{Ker} (c)$ such that $\zeta^\top a\ne0$. Then, for $\lambda\in\mathbb R$,
    \[
        (\lambda \, \zeta)^\top a-\frac12(\lambda \, \zeta)^\top c \,(\lambda \, \zeta)=\lambda \, \zeta^\top a,
    \]
    which is unbounded above as $\lambda$ varies. Hence, the supremum is infinite.
    
    Suppose now that $a\in\operatorname{Range}(c)$. Then $cc^\dagger a=a$. For any $p\in\mathbb R^d$,
    \[
        p^\top a-\frac12p^\top c\,p=\frac12a^\top c^\dagger a-\frac12\big(p-c^\dagger a\big)^\top c\big(p-c^\dagger a\big).
    \]
    The last term is nonpositive, and equality holds precisely when $c(p-c^\dagger a)=0$, that is, when $p=c^\dagger a+\zeta$ for some $\zeta\in \operatorname{Ker} (c)$.
\end{proof}

Back to our setup, this result provides the following structural characterization of local  finiteness of growth in the sense of Definition~\ref{def:finite-resolution-locally-finite-growth}.

\begin{proposition}[Structural Characterization of Locally Finite Growth] \label{prop:structural-characterization-lfg}
    The return path $R \in C([-h,T];\R^d)$ has locally finite growth on $[0,T]$ relative to the triple $(\mathcal D, h, \Pi)$ if, and only if, $a^{h,\mathcal D}(t)\in\operatorname{Range}\big(c^{h,\mathcal D}(t)\big)$, $\d O^{h,\mathcal D}$-a.e. on $[0,T]$, and
    \[
        \int_0^T \big(a^{h,\mathcal D}(t)\big)^\top\big(c^{h,\mathcal D}(t)\big)^\dagger a^{h,\mathcal D}(t)\, \d O^{h,\mathcal D}(t) < \infty.
    \]
    In this case, the maximal growth rate of \eqref{def:maximal-growth-rate} is given by 
    \[
        g^{h,\mathcal D}(t) = \frac12 \big(a^{h, \mathcal D}(t)\big)^\top \big(c^{h,\mathcal D}(t)\big)^\dagger a^{h, \mathcal D}(t), \qquad \d O^{h,\mathcal D}\text{-a.e.},
    \]
    and   the cumulative maximal growth of \eqref{def:cumulative-maximal-growth} by
    \[
        G^{h, \mathcal D}(t) = \frac12 \int_0^t \big(a^{h,\mathcal D}(s)\big)^\top \big(c^{h,\mathcal D}(s)\big)^\dagger  a^{h,\mathcal D}(s)\, \d O^{h,\mathcal D}(s), \qquad 0 \le t \le T.
    \]
\end{proposition}

\begin{proof}
    If \(G^{h,\mathcal D}(T)<\infty\), then \(g^{h,\mathcal D}<\infty\)\,, ~ \( \d O^{h,\mathcal D}\)-a.e.; Lemma~\ref{lem:pointwise-maximal-growth} then gives the range condition and the formula for \(g^{h,\mathcal D}\). The converse follows by integrating that formula.
\end{proof}

Under local finiteness of growth as in Definition~\ref{def:finite-resolution-locally-finite-growth}, the pseudo-inverse provides a ``canonical'' pointwise maximizer for \eqref{def:maximal-growth-rate} in the class   $\mathfrak{P}^{h, \mathcal D}$, namely,
\begin{equation}    \label{def:canonical-maximizer}
    \nu^{h,\mathcal D}(t) := \big(c^{h,\mathcal D}(t)\big)^\dagger a^{h,\mathcal D}(t), \qquad 0 \le t \le T.
\end{equation}
In the notation of \eqref{def-gamma-pi}, \eqref{def-Gamma}, \eqref{def:maximal-growth-rate} we have then
\begin{equation}    \label{eq:gamma_nu=g}
    \gamma^{h,\mathcal D}_{\nu^{h,\mathcal D}} \equiv g^{h,\mathcal D}, \quad \d O^{h, \mathcal D}\text{-a.e.}; \qquad \Gamma^{h, \mathcal D}_{\nu^{h,\mathcal D}}(t) \equiv G^{h,\mathcal D}(t), \qquad 0 \le t \le T.
\end{equation}

\begin{remark}
    The maximization of pointwise growth rate is unique only up to $c^{h,\mathcal D}$-null portfolios. More precisely, if a portfolio $\widetilde\nu \in \mathfrak P^{h, \mathcal D}$ satisfies
    \[
        c^{h,\mathcal D}(t) \, \widetilde\nu(t) = a^{h,\mathcal D}(t)\,, \quad \d O^{h,\mathcal D}\text{-a.e.},
    \]
    then $\widetilde\nu$ is also a pointwise maximizer, and
    \[
        \widetilde\nu(t) - \nu^{h, \mathcal D}(t) \in \operatorname{Ker} \big(c^{h,\mathcal D}(t)\big)\,, \quad \d O^{h,\mathcal D}\text{-a.e.}
    \]
    Every such maximizer has the same cumulative growth $\Gamma^{h,\mathcal D}_{\widetilde\nu} \equiv G^{h,\mathcal D} $ as $\nu^{h, \mathcal D} $; and the same covariation characteristic with every portfolio $\pi   \in \mathfrak P^{h, \mathcal D}$, namely, $C^{h,\mathcal D}_{\widetilde\nu,\pi}  \equiv C^{h,\mathcal D}_{\nu^{h,\mathcal D},\pi}$.
\end{remark} 

\medskip

\subsection{Growth optimality}    \label{subsec:finite-resolution-growth-optimality}

We develop now consequences of the local finiteness of growth condition $G^{h,\mathcal D}(T)<\infty$ in Definition~\ref{def:finite-resolution-locally-finite-growth}, in terms of the pathwise local characteristic $(a^{h,\mathcal D},c^{h,\mathcal D})$ of Definition~\ref{def:pathwise-local-characteristics-extractor} and of the growth characteristics in Subsection~\ref{subsec:portfolios-growth}.

\begin{definition}[Growth-Optimal Portfolio]  \label{def:finite-resolution-growth-optimal}
    A portfolio $\rho\in\mathfrak P^{h,\mathcal D}$ is called \emph{$(h,\mathcal D)$-growth-optimal}, if the relative cumulative growth $\Gamma^{h,\mathcal D}_{\pi|\rho} = \Gamma^{h,\mathcal D}_\pi-\Gamma^{h,\mathcal D}_\rho$ of \eqref{def:relative-cumulative-growth} is non-increasing on $[0,T]$ for every $\pi\in\mathfrak P^{h,\mathcal D}$.
\end{definition}

\begin{proposition}[Growth Optimality and Structural Condition]  \label{prop:finite-resolution-growth-structural-benchmark}
    In the setting of Subsections~\ref{subsec:nonanticipative-decomposition} and \ref{subsec:portfolios-growth}, the following conditions are equivalent:
    \begin{enumerate}[label=(\roman*)]
        \item The return path $R$ of \eqref{def:return} has locally finite growth on $[0,T]$; that is, $G^{h,\mathcal D}(T)<\infty$ as in Definition~\ref{def:finite-resolution-locally-finite-growth}.
        
        \item The structural conditions of Proposition~\ref{prop:structural-characterization-lfg} hold; namely,
        \[
            a^{h,\mathcal D}(t) \in \operatorname{Range} \big(c^{h,\mathcal D}(t)\big)\,, \qquad \d O^{h,\mathcal D}\text{-a.e. on }[0,T],
        \]
        \[
            \int_0^T \big(a^{h,\mathcal D}(t)\big)^\top \big(c^{h,\mathcal D}(t)\big)^\dagger a^{h,\mathcal D}(t)\,\d O^{h,\mathcal D}(t) < \infty.
        \]
        
        \item There exists a portfolio $\nu\in\mathfrak P^{h,\mathcal D}$ such that
        \begin{equation}    \label{eq:structural-condition}
            a^{h,\mathcal D}(t) = c^{h,\mathcal D}(t)\,\nu(t)\,, \qquad \d O^{h,\mathcal D}\text{-a.e. on }[0,T].
        \end{equation}
        
        \item There exists an $(h, \mathcal D)$-growth-optimal portfolio.
        
        \item There exists a portfolio $\nu\in\mathfrak P^{h,\mathcal D}$ such that, for every portfolio $\pi\in\mathfrak P^{h,\mathcal D}$, we have 
        \[
            \Gamma^{h,\mathcal D}_\nu(t)-\Gamma^{h,\mathcal D}_\pi(t) = \frac12 \,C^{h,\mathcal D}_{\pi-\nu}(t), \qquad 0\le t\le T\,.
        \]
    \end{enumerate}
\end{proposition}

\begin{remark}
    Under any (therefore, all) of the above conditions, every portfolio with the property (iii) is given, up to $c^{h,\mathcal D}$-null portfolios, by
    \[
        \nu^{h,\mathcal D}(t) = \big(c^{h,\mathcal D}(t)\big)^\dagger a^{h,\mathcal D}(t), \qquad \d O^{h,\mathcal D}\text{-a.e.},
    \]
    as in \eqref{def:canonical-maximizer}; and 
    \[
        \Gamma^{h,\mathcal D}_{\nu^{h,\mathcal D}}(t) = G^{h,\mathcal D}(t), \qquad 0 \le t \le T
    \]
    holds in the manner of \eqref{eq:gamma_nu=g}. Moreover, for every portfolio $\pi\in\mathfrak P^{h,\mathcal D}$, Lemma~\ref{lem:relative-growth-identity} gives
    \begin{equation*}
        \Gamma^{h,\mathcal D}_{\nu^{h,\mathcal D}|\pi}(t) = \Gamma^{h, \mathcal D}_{\nu^{h, \mathcal D}}(t) - \Gamma^{h, \mathcal D}_{\pi}(t) = \frac12 \,C^{h,\mathcal D}_{\pi-\nu^{h,\mathcal D}}(t), \qquad 0 \le t \le T\,.
    \end{equation*}
\end{remark}

\begin{proof}
    The equivalence of {\rm(i)} and {\rm(ii)} follows from Proposition~\ref{prop:structural-characterization-lfg}.
    
    Now let us assume {\rm(ii)}, and consider $\nu^{h,\mathcal D}(t)$ as in \eqref{def:canonical-maximizer}. Since $a^{h,\mathcal D}\in\operatorname{Range}(c^{h,\mathcal D})\,,$ $\d O^{h,\mathcal D}$-a.e., we have
    \[
        c^{h,\mathcal D}\nu^{h,\mathcal D} = c^{h,\mathcal D} \big(c^{h,\mathcal D}\big)^\dagger a^{h,\mathcal D} = a^{h,\mathcal D}\,, \qquad \d O^{h,\mathcal D}\text{-a.e.}
    \]
    Furthermore,
    \[
        \int_0^T \big|\big(\nu^{h,\mathcal D}(t)\big)^\top a^{h,\mathcal D}(t)\big|\,\d O^{h,\mathcal D}(t) = \int_0^T \big( a^{h,\mathcal D}(t)\big)^\top \big(c^{h,\mathcal D}(t)\big)^\dagger  a^{h,\mathcal D}(t)\,\d O^{h,\mathcal D}(t) < \infty
    \]
    holds, as well as
    \[
        \int_0^T \big(\nu^{h,\mathcal D}(t)\big)^\top c^{h,\mathcal D}(t)\nu^{h,\mathcal D}(t)\, \d O^{h,\mathcal D}(t) = \int_0^T \big( a^{h,\mathcal D}(t)\big)^\top \big(c^{h,\mathcal D}(t)\big)^\dagger a^{h,\mathcal D}(t)\, \d O^{h,\mathcal D}(t) < \infty.
    \]
    Thus $\nu^{h,\mathcal D}\in\mathfrak P^{h,\mathcal D}$ and condition {\rm(iii)} holds.
    
    Conversely, assume {\rm(iii)} and let $\nu\in\mathfrak P^{h,\mathcal D}$ satisfy \eqref{eq:structural-condition}. Then $ a^{h,\mathcal D} \in \operatorname{Range} (c^{h,\mathcal D})\,,$ $~\d O^{h,\mathcal D}$-a.e. Moreover,
    \[
        \big( a^{h,\mathcal D}\big)^\top \big(c^{h,\mathcal D}\big)^\dagger a^{h,\mathcal D} = \nu^\top c^{h,\mathcal D}\nu\,, \qquad \d O^{h,\mathcal D}\text{-a.e.}
    \]
    Since $\nu\in\mathfrak P^{h,\mathcal D}$, the right-hand side is $O^{h,\mathcal D}$-integrable, and thus condition {\rm(ii)} holds. Hence, {\rm(ii)} and {\rm(iii)} are equivalent.
    
    Next, suppose again that {\rm(iii)} holds. Applying Lemma~\ref{lem:relative-growth-identity} with $\rho=\nu$ gives, for every $\pi\in\mathfrak P^{h,\mathcal D}$,
    \[
        \Gamma^{h,\mathcal D}_{\pi|\nu}(t) = \int_0^t \big(\pi(s)-\nu(s)\big)^\top \big( a^{h,\mathcal D}(s)-c^{h,\mathcal D}(s)\nu(s)\big)\,\d O^{h,\mathcal D}(s) - \frac12 \,C^{h,\mathcal D}_{\pi-\nu}(t).
    \]
    The first term on the right-hand side vanishes, hence
    \[
        \Gamma^{h,\mathcal D}_{\pi|\nu}(t) = \Gamma^{h,\mathcal D}_\pi(t) - \Gamma^{h,\mathcal D}_\nu(t) = -\frac12 \,C^{h,\mathcal D}_{\pi-\nu}(t), \qquad 0 \le t \le T.
    \]
    This proves {\rm(v)}. Since $C^{h,\mathcal D}_{\pi-\nu}$ is nondecreasing, the identity also implies that $\Gamma^{h,\mathcal D}_{\pi|\nu}$ is nonincreasing for every $\pi\in\mathfrak P^{h,\mathcal D}$, and {\rm(iv)} follows. Conversely, condition {\rm(v)} immediately implies {\rm(iv)}.
    
    \smallskip
    It remains to prove that {\rm(iv)} implies {\rm(iii)}. Let $\rho\in\mathfrak P^{h,\mathcal D}$ be $(h, \mathcal D)$-growth-optimal. Fix a bounded Borel map $\, \xi:[0,T]\to\mathbb R^d \,$ and $\varepsilon>0$. Because $\mathfrak P^{h,\mathcal D}$ is stable under bounded perturbations, $\rho+\varepsilon\xi\in\mathfrak P^{h,\mathcal D}$; and by Lemma~\ref{lem:relative-growth-identity},
    \[
        \Gamma^{h,\mathcal D}_{\rho+\varepsilon\xi|\rho}(t)
        = \int_0^t \Big( \varepsilon \xi^\top(s) \big( a^{h,\mathcal D}(s) - c^{h,\mathcal D}(s)\rho(s)\big) - \frac{\varepsilon^2}{2}\xi^\top(s) c^{h,\mathcal D}(s) \xi(s) \Big) \,\d O^{h,\mathcal D}(s).
    \]
    Since $\rho$ is $(h, \mathcal D)$-growth-optimal, the left-hand side is nonincreasing in $t$. Therefore, we have 
    \[
        \varepsilon\xi^\top (a^{h,\mathcal D}-c^{h,\mathcal D}\rho) - \frac{\varepsilon^2}{2}\xi^\top c^{h,\mathcal D}\xi \le 0\,, \quad \d O^{h,\mathcal D}\text{-a.e.}
    \]
    Dividing by $\varepsilon$ and sending $\varepsilon\downarrow0$ yields $\xi^\top(a^{h,\mathcal D}-c^{h,\mathcal D}\rho) \le 0$; whereas, replacing $\xi$ by $-\xi$ gives $\xi^\top(a^{h,\mathcal D}-c^{h,\mathcal D}\rho) = 0$. Since $\xi$ was arbitrary, taking $\xi$ successively equal to the coordinate vectors yields $a^{h,\mathcal D}=c^{h,\mathcal D}\rho$, $\d O^{h,\mathcal D}$-a.e.
as claimed. \end{proof}

\medskip

\subsection{Wealth-level interpretation}  \label{subsec:wealth-level}

We need to justify the ``portfolio'' terminology introduced in Subsection~\ref{subsec:portfolios-growth}. We shall do this by pairing functions $\pi: [0, T] \to \R^d$, in a suitable subclass $\, \mathfrak W^{h, \mathcal D}$ of $\, \mathfrak P^{h, \mathcal D}$ as in Definition~\ref{def:portfolios}, with the strictly positive ``wealth levels'' $X_{\pi}(t), 0 \le t \le T$ they generate starting with  capital $X_{\pi}(0)=1$. The interpretation here, is that $\pi_i(t)$ stands for the proportion of the wealth $X_{\pi}(t)>0$ that gets invested at time $t\in[0,T]$ in the $i$-th asset, whose return path $R_i(t), 0 \le t \le T$ is the $i$-th component of $R$ in \eqref{def:return}.

In order for this effort to be successful, we need to select the subclass $\mathfrak W^{h, \mathcal D} \subset \mathfrak P^{h, \mathcal D}$ in a manner that allows certain F\"ollmer integrals, as in Lemma~\ref{lem:follmer-ito}, to be well-defined. More precisely, we need the following notion.

\begin{definition}[Basic Admissible F\"ollmer Integrands]   \label{def:basic-admissible-follmer-integrands}
    A map $\pi:[0,T]\to\mathbb R^d$ is called \emph{basic admissible F\"ollmer integrand for $R$}, if there exist an integer $m\ge1$, a continuous finite-variation path $B:[0,T]\to\mathbb R^m$, and a function $f\in C^{1,2}(\mathbb R^m\times\mathbb R^d)$, such that
    \begin{equation}    \label{eq:basic-admissible-representation}
        \pi(t)=\nabla_x f(B(t),R(t)),   \qquad 0\le t\le T.
    \end{equation}
    The collection of all such integrands is denoted by $\mathfrak A_\Pi(R)$.
\end{definition}

We stress again, in the spirit of Remark~\ref{rem:Ito-integrands}, the {\it very} special character of the recipe \eqref{eq:basic-admissible-representation}, as gradient of a $C^{1,2}$ function evaluated at the current values $(B(t), R(t))$ of the paths $B, R$.

\begin{definition}[Wealth-generating F\"ollmer Portfolios]  \label{def:follmer-integrable-portfolios}
    Let $\mathbb D_\Pi:=\bigcup_{n\ge1}\Pi_n$.  The class $\mathfrak W^{h,\mathcal D}$ of \emph{wealth-generating F\"ollmer portfolios} is the finite $\mathbb D_\Pi$-pasting and finite-point modification closure of $\mathfrak A_\Pi(R)$ in Definition~\ref{def:basic-admissible-follmer-integrands}.  Explicitly, $\pi\in\mathfrak W^{h,\mathcal D}$ if there exist $0=u_0<u_1<\cdots<u_L=T$, $u_1,\ldots,u_{L-1}\in\mathbb D_\Pi$, a finite set $F\subset\mathbb D_\Pi$, and $\pi^1,\ldots,\pi^L$ in the collection $\mathfrak A_\Pi(R)$ of Definition~\ref{def:basic-admissible-follmer-integrands}, such that, for every $t\notin F$,
    \begin{equation*}
        \pi(t)=\sum_{\ell=1}^L\mathbf1_{I_\ell}(t)\pi^\ell(t), \qquad I_\ell:=
        \begin{cases}
            [u_{\ell-1},u_\ell),&\ell<L,
            \\ [u_{L-1},T],&\ell=L.
        \end{cases}
    \end{equation*}
    The convention at a switching time assigns the increment immediately after that time to the new portfolio. The values on $F$ are unrestricted. For a continuous integrator, changing an integrand at finitely-many fixed times does not alter its F\"ollmer integral, because the affected increments vanish with the mesh.
\end{definition}

\begin{proposition}[Calculus for Wealth-generating Portfolios]  \label{prop:wealth-generating-follmer-calculus}
    The class $\mathfrak W^{h,\mathcal D}$ is a linear subspace of $\mathfrak P^{h,\mathcal D}$, contains every constant portfolio, and is closed under finite pasting at times in $\mathbb D_\Pi$. For every $\pi,\rho\in\mathfrak W^{h,\mathcal D}$, the F\"ollmer integrals
    \begin{equation}    \label{def:R-pi-M-pi}
        R_\pi(t):=\int_0^t\pi(s)^\top\,\d^\Pi R(s), \qquad M_\pi^{h,\mathcal D}(t):=\int_0^t\pi(s)^\top\,\d^\Pi M^{h,\mathcal D}(s)
    \end{equation}
    exist as uniform limits on $[0,T]$, and with the notation of \eqref{def:cumulative-return-covariation}, \eqref{def:R-pi-M-pi}:
    \begin{equation}    \label{eq:R-pi-decomposition}
        R_\pi=A_\pi^{h,\mathcal D}+M_\pi^{h,\mathcal D}, \qquad
        [R_\pi,R_\rho]_\Pi = [M_\pi^{h,\mathcal D},M_\rho^{h,\mathcal D}]_\Pi = C_{\pi,\rho}^{h,\mathcal D}.
    \end{equation}
\end{proposition}

\begin{proof}
    For a basic F\"ollmer integrand as in Definition~\ref{def:basic-admissible-follmer-integrands}, existence of the integral with respect to $R$ follows from Lemma~\ref{lem:follmer-ito-fv-control}. Since $R=R(0)+A^{h,\mathcal D}+M^{h,\mathcal D}$, the same integrand is admissible for $M^{h,\mathcal D}$ after enlarging the finite-variation control: if $\pi=\nabla_x f(B,R)$, set
    \[
        \widetilde B:=(B,A^{h,\mathcal D}), \qquad \widetilde f(b,a,m):=f(b,R(0)+a+m).
    \]
    Then $\pi=\nabla_m\widetilde f(\widetilde B,M^{h,\mathcal D})$. Linearity of the Riemann sums and classical integration against the finite-variation path $A^{h,\mathcal D}$ give the first identity in \eqref{eq:R-pi-decomposition}. The quadratic-covariation transformation rule \cite[Proposition~12]{schied2014} for admissible F\"ollmer integrands gives the second identity of \eqref{eq:R-pi-decomposition}. All these statements extend to finite $\mathbb D_\Pi$-pastings because every sufficiently fine partition contains each switching time and the sums then split exactly at those times. A finite-point modification changes each Riemann sum in only finitely-many terms; uniform continuity of $R$ and $M^{h,\mathcal D}$ makes their total contribution converge uniformly to zero.
    
    A basic admissible integrand is continuous and hence bounded on $[0,T]$; a finite pasting is bounded as well.  The definition of the clock, therefore, implies membership in $\mathfrak P^{h,\mathcal D}$. Linearity follows by combining the finite-variation controls and adding the generating functions.
\end{proof}

We introduce now, for any given $Y \in Q_\Pi([0,T])$ with $Y(0)=0$ as in Definition~\ref{def:Q-Pi}, the \emph{pathwise Dol\'eans exponential}
\[
    \mathcal E^\Pi(Y)(t):=\exp\Big(Y(t)-\frac12\,[Y]_\Pi(t)\Big),\qquad 0 \le t \le T.
\]
Whenever a continuous function $Z:[0,T]\to(0,\infty)$ can be represented as $Z=\mathcal E^\Pi(Y)$ for some $Y\in Q_\Pi([0,T])$ with $Y(0)=0$, we call $Y$ \emph{pathwise Dol\'eans logarithm} of $Z$, and write $Y=\mathcal L^\Pi(Z)$. This path is unique and admits the explicit representation
\[
    \mathcal L^\Pi(Z) = \log Z+\frac12[\log Z]_\Pi.
\]
Indeed, if $Z=\mathcal E^\Pi(Y)$, then $\log Z=Y-\frac12[Y]_\Pi$, and Lemma~\ref{lem:BV-zero-QV} gives $[\log Z]_\Pi=[Y]_\Pi$.

\begin{definition}[Wealth Generated by F\"ollmer-Integrable Portfolios]\label{def:wealth-follmer-integrable-portfolio}
    For any given $\pi\in\mathfrak W^{h,\mathcal D}$ as in Definition~\ref{def:follmer-integrable-portfolios}, the \emph{wealth} $X_{\pi}(\cdot)$ generated by $\pi$ starting with capital $X_\pi(0)=1$, is
    \begin{equation}    \label{def:wealth}
        X_\pi(t):=\mathcal E^\Pi(R_\pi)(t)=\exp\Big(R_\pi(t)-\frac12 \, C_\pi^{h,\mathcal D}(t)\Big), \qquad 0 \le t \le T.
    \end{equation}
\end{definition}

From the decomposition $R_\pi=A_\pi^{h,\mathcal D}+M_\pi^{h,\mathcal D}$ in \eqref{eq:R-pi-decomposition}, and \eqref{def-Gamma}, the log-wealth is given as
\begin{equation}    \label{eq:log-X-pi}
    \log X_\pi(t) = \Gamma_\pi^{h,\mathcal D}(t) + M_\pi^{h,\mathcal D}(t),\qquad 0 \le t \le T.
\end{equation}
This expression highlights $\Gamma^{h,\mathcal D}_{\pi}$ as the finite-variation growth characteristic of logarithmic wealth, and $M^{h,\mathcal D}_{\pi}$ as the residual F\"ollmer integral component.

\begin{proposition}[Self-financing Identity]    \label{prop:self-financing}
    For every portfolio $\pi\in\mathfrak W^{h,\mathcal D}$, both F\"ollmer integrals in the following self-financing identity
    \begin{equation}\label{eq:pathwise-self-financing}
        X_\pi(t)=1+\int_0^t X_\pi(s)\,\d^\Pi R_\pi(s)
        =1+\int_0^t X_\pi(s)\pi(s)^\top\,\d^\Pi R(s), \qquad 0\le t\le T,
    \end{equation}
    are well-defined. Moreover, $X_\pi$ is the unique solution, in the class of admissible F\"ollmer integrands for the scalar path $R_\pi$, of the first integral equation in \eqref{eq:pathwise-self-financing} with initial value $X_\pi(0)=1$.
\end{proposition}

\begin{proof}
    Fix $\pi\in\mathfrak W^{h,\mathcal D}$ and set $Y:=R_\pi$, $Q:=[Y]_\Pi=C_\pi^{h,\mathcal D}$, where the last identity follows from Proposition~\ref{prop:wealth-generating-follmer-calculus}. Thus $Y\in Q_\Pi([0,T])$, the path $Q$ is continuous and of finite variation, and $Y(0)=Q(0)=0$.

    Consider the function $F(q,y):=\exp(y-q/2)$ for $(q,y)\in\mathbb R^2$. Then $F\in C^{1,2}(\mathbb R^2)$ and by Definition~\ref{def:wealth-follmer-integrable-portfolio}, $X_\pi(t)=F\big(Q(t),Y(t)\big) = \partial_yF\big(Q(t),Y(t)\big)$. Hence $X_\pi$ is a basic admissible F\"ollmer integrand for the scalar integrator $Y$, in the sense of Definition~\ref{def:basic-admissible-follmer-integrands} with $d=1$.

    Applying Lemma~\ref{lem:follmer-ito-fv-control} with the finite variation control $Q$ and the integrator $Y$ gives
    \begin{align*}
        X_\pi(t)-1 &= -\frac12\int_0^tX_\pi(s)\,\d Q(s) +\int_0^tX_\pi(s)\,\d^\Pi Y(s) +\frac12\int_0^tX_\pi(s)\,\d[Y]_\Pi(s)
        \\ &= \int_0^tX_\pi(s)\,\d^\Pi Y(s),
    \end{align*}
    because $Q=[Y]_\Pi$. Since $Y=R_\pi$, this proves the first equality in \eqref{eq:pathwise-self-financing}.

    The definition \eqref{def:R-pi-M-pi} of $R_\pi$, together with the associativity rule for the F\"ollmer integral \cite[Theorem~13]{schied2014}, therefore gives
    \[
        \int_0^tX_\pi(s)\,\d^\Pi R_\pi(s) = \int_0^tX_\pi(s)\pi(s)^\top\,\d^\Pi R(s).
    \]
    For a portfolio obtained by finite $\mathbb D_\Pi$-pasting, the same conclusion follows by applying associativity on each basic segment: every sufficiently fine partition contains all switching times, so the corresponding Riemann sums split at those times. Finite-point modifications do not affect the limits because the integrators are continuous. This proves the second equality in \eqref{eq:pathwise-self-financing}.

    Finally, uniqueness within the class of admissible F\"ollmer integrands follows from the uniqueness result for the homogeneous linear F\"ollmer integral equation in \cite[Proposition~3.1]{Hirai}, applied to the scalar integrator $Y=R_\pi$.
\end{proof}

\smallskip

\noindent \emph{Interpretation}: The equation \eqref{eq:pathwise-self-financing} highlights the significance of $X_{\pi}(\cdot)$ as the ``wealth'' generated by investing in a financial market, whose $d$ assets have corresponding cumulative returns $R_1, \cdots, R_d$, according to the portfolio weights $\,\pi_1, \cdots, \pi_d\,$. Note that these latter need not be positive or satisfy $\sum_{i=1}^d \pi_i(\cdot) \equiv 1$; meaning, that ``short-selling'' of assets, and ``keeping cash under the mattress'' (i.e., access to a money-market account with zero interest), are both allowed.

\medskip

We express now the relative performance of a given portfolio $\pi \in \mathfrak W^{h, \mathcal D}$, when measured relative to that of a ``baseline'' portfolio $\rho\in\mathfrak W^{h, \mathcal D}$. In what follows, we define for $\pi,\rho\in\mathfrak P^{h,\mathcal D}:$
\begin{equation}   \label{def:B-relative-drift}
    B_{\pi|\rho}^{h,\mathcal D}(t) := A_{\pi-\rho}^{h,\mathcal D}(t) - C_{\pi-\rho,\rho}^{h,\mathcal D}(t) = \int_0^t \big(\pi(s)-\rho(s)\big)^\top \big( a^{h,\mathcal D}(s)-c^{h,\mathcal D}(s)\rho(s) \big)\,\d O^{h,\mathcal D}(s).
\end{equation}

\begin{proposition}[Relative Wealth Decomposition]  \label{prop:relative-wealth-decomposition}
    For any two given portfolios $\pi,\rho\in\mathfrak W^{h,\mathcal D}$ we have
    \[
        \log\frac{X_\pi(t)}{X_\rho(t)} = B^{h,\mathcal D}_{\pi|\rho}(t) - \frac12 C_{\pi-\rho}^{h,\mathcal D}(t) + M^{h, \mathcal D}_{\pi-\rho}(t), \qquad 0 \le t \le T
    \]
    or equivalently
    \[
        \frac{X_\pi}{X_\rho} = \mathcal E^\Pi\big(B^{h,\mathcal D}_{\pi|\rho}+M_{\pi-\rho}^{h,\mathcal D}\big), \qquad \mathcal L^\Pi\Big(\frac{X_\pi}{X_\rho}\Big) = B^{h,\mathcal D}_{\pi|\rho}+M_{\pi-\rho}^{h,\mathcal D}\,.
    \]
\end{proposition}

\begin{proof}
    By Definition~\ref{def:wealth-follmer-integrable-portfolio},
    \[
        \log\frac{X_\pi(t)}{X_\rho(t)} = R_\pi(t)-R_\rho(t)-\frac12\big(C_\pi^{h,\mathcal D}(t)-C_\rho^{h,\mathcal D}(t)\big).
    \]
    Using $R_\eta=A_\eta^{h,\mathcal D}+M_\eta^{h,\mathcal D}$ for $\eta=\pi,\rho$, we get $R_\pi-R_\rho=A_{\pi-\rho}^{h,\mathcal D}+M_{\pi-\rho}^{h,\mathcal D}$. Moreover, the bilinearity of $C_{\cdot,\cdot}^{h,\mathcal D}$ gives
    \[
        C_\pi^{h,\mathcal D}-C_\rho^{h,\mathcal D}=C_{\pi-\rho}^{h,\mathcal D}+2C_{\pi-\rho,\rho}^{h,\mathcal D}\,.
    \]
    Therefore
    \[
        \log\frac{X_\pi(t)}{X_\rho(t)} = M_{\pi-\rho}^{h,\mathcal D}(t)+A_{\pi-\rho}^{h,\mathcal D}(t)-C_{\pi-\rho,\rho}^{h,\mathcal D}(t)-\frac12 \,C_{\pi-\rho}^{h,\mathcal D}(t),
    \]
    which proves the logarithmic identity. Since $B^{h,\mathcal D}_{\pi|\rho}$ is of finite variation, it has zero quadratic variation and zero covariation with $M_{\pi-\rho}^{h,\mathcal D}\,$; therefore, 
    \[
        \big[M_{\pi-\rho}^{h,\mathcal D}+B^{h,\mathcal D}_{\pi|\rho}\big]_{\Pi} = \big[M_{\pi-\rho}^{h,\mathcal D}\big]_{\Pi} = C_{\pi-\rho}^{h,\mathcal D} 
    \]
as well as
    \[
        \frac{X_\pi}{X_\rho} = \exp\Big(M_{\pi-\rho}^{h,\mathcal D}+B^{h,\mathcal D}_{\pi|\rho}-\frac12 C_{\pi-\rho}^{h,\mathcal D}\Big) 
        = \mathcal E^\Pi\Big(M_{\pi-\rho}^{h,\mathcal D}+B^{h,\mathcal D}_{\pi|\rho}\Big).
    \]
\end{proof}

\begin{proposition} [Driftless Num\'eraire Representation]  \label{prop:driftless-numeraire}
    A portfolio $\nu\in\mathfrak W^{h,\mathcal D}$ satisfies 
    \[
        a^{h,\mathcal D}(t) = c^{h,\mathcal D}(t) \, \nu(t), \qquad \d O^{h,\mathcal D}\text{-a.e. } t \in [0,T]
    \]
    as in \eqref{eq:structural-condition} if, and only if, it has in the notation of \eqref{def:B-relative-drift} the \emph{driftless num\'eraire property}
    \begin{equation}    \label{def:driftless-property}
        B^{h,\mathcal D}_{\pi|\nu} \equiv 0, \qquad \forall \, \pi\in\mathfrak P^{h,\mathcal D}.
    \end{equation}
    In this case, for every $\pi \in \mathfrak W^{h, \mathcal D}$,
    \[
        \frac{X_\pi}{X_\nu} = \mathcal E^\Pi\big(M_{\pi-\nu}^{h,\mathcal D}\big), \qquad \log\frac{X_\pi(t)}{X_\nu(t)} = -\frac12 \, C_{\pi-\nu}^{h,\mathcal D}(t) + M_{\pi-\nu}^{h,\mathcal D}(t), \qquad 0 \le t \le T.
    \]
\end{proposition}

\begin{proof}
    By \eqref{def:B-relative-drift}, we have 
    \[
        B^{h,\mathcal D}_{\pi|\nu}(t) = \int_0^t \big(\pi(s)-\nu(s)\big)^\top \big(a^{h, \mathcal D}(s) -c^{h, \mathcal D}(s) \nu(s)\big)\,\d O^{h, \mathcal D}(s), \qquad 0 \le t \le T.
    \]
    If $a^{h, \mathcal D} = c^{h, \mathcal D}\nu\,,$ $\d O^{h, \mathcal D}$-a.e., then $B^{h,\mathcal D}_{\pi|\nu}\equiv0$ for every $\pi\in\mathfrak P^{h, \mathcal D}$.

    Conversely, suppose that $B^{h,\mathcal D}_{\pi|\nu}\equiv0$ for every $\pi\in\mathfrak P^{h, \mathcal D}$. Since $\mathfrak P^{h, \mathcal D}$ is stable under bounded Borel perturbations, we have $\nu+\xi\in\mathfrak P^{h, \mathcal D}$ for every bounded Borel map $\xi:[0,T]\to\mathbb R^d$. Hence
    \[
        0 = B^{h,\mathcal D}_{\nu+\xi|\nu}(t) = \int_0^t \xi^\top(s) \big(a^{h, \mathcal D}(s)-c^{h, \mathcal D}(s)\,\nu(s)\big)\,\d O^{h, \mathcal D}(s), \qquad 0 \le t \le T.
    \]
    Taking $\xi=e_i\mathbf 1_E$, where $e_i$ is the $i$-th unit vector and $E\subset[0,T]$ is a Borel set, gives
    \[
        \int_E \Big(a_i^{h,\mathcal D}(s)-\big(c^{h,\mathcal D}(s) \nu(s)\big)_i\Big)\,\d O^{h,\mathcal D}(s) = 0.
    \]
    Since this holds for every $i=1,\ldots,d$ and every Borel set $E$, we obtain
    $\,
        a^{h, \mathcal D}=c^{h, \mathcal D} \nu \,,~ \d O^{h, \mathcal D}\text{-a.e.}
   $
    The last claim follows from Proposition~\ref{prop:relative-wealth-decomposition}.
\end{proof}

\begin{proposition}[Wealth-Level Version of the Structural Condition]   \label{prop:structural-condition-driftless-numeraire}
    Assume that the class $\mathfrak W^{h, \mathcal D}$ of F\"ollmer-integrable portfolios in Definition \ref{def:follmer-integrable-portfolios}, contains the portfolio $\nu^{h,\mathcal D}$ of \eqref{def:canonical-maximizer}. Then, the equivalent conditions of Proposition~\ref{prop:finite-resolution-growth-structural-benchmark}, are also equivalent to the existence of a portfolio in $\,\mathfrak W^{h,\mathcal D}$ with the driftless num\'eraire property \eqref{def:driftless-property}.
\end{proposition}

\begin{proof}
    Follows directly from Propositions~\ref{prop:finite-resolution-growth-structural-benchmark}, \ref{prop:driftless-numeraire}.
\end{proof}

\medskip

\section{Scenario-wise pathwise viability}\label{sec:scenario-wise-viability}

We formulate now, in a scenario-wise setting, a pathwise analogue of the notion of ``viability" from \cite{KaratzasKardaras2021}. This proscribes ``very egregious arbitrage'', namely, the possibility of financing a non-trivial future liability stream starting with arbitrarily small initial capital. The single-path construction developed in the previous section provides pathwise local characteristics and self-financing wealth processes along every fixed return path, and characterizes driftless num\'eraire portfolios whenever the structural condition holds and the corresponding portfolio belongs to the wealth-generating class.

Viability as in \cite{KaratzasKardaras2021}, however, is {\it not}  a single-path concept: it concerns the possibility of financing nontrivial withdrawal streams using non-anticipative trading rules, {\it uniformly over a class of possible future scenarios}. Thus, we introduce here a scenario-set of return paths, and  introduce the notion of robust pathwise financing capital.  

\medskip

\subsection{Scenario sets and non-anticipative trading rules}   \label{subsec:scenario-sets}

We fix a window length $h>0$, a trading horizon $T>0$, a refining sequence of partitions $\Pi=(\Pi_n)_{n \in \mathbb N}$ of $[0,T]$, and a non-anticipative trend extractor ${\cal D}$. With the notation 
\begin{equation}    \label{def:Omega-Pi}
    \Omega_\Pi:=\Big\{\omega\in C([-h,T];\mathbb R^d):\omega \big|_{[0,T]}\in Q_\Pi([0,T];\mathbb R^d)\Big\}
\end{equation}
as in Definition~\ref{def:Q-Pi}, a ``scenario set" is a subset $\,\Omega\subset\Omega_\Pi$. The coordinate return path $R( \omega)$ is
\[    
    R(s, \omega):=\omega(s),\qquad s\in[-h,T],\quad \omega\in\Omega\,,
\]
and the scenario-set $\Omega$ is equipped with the canonical information flow
\[
    \mathcal F^\Omega (t):=\sigma\big(R(s):-h\le s\le t\big),\qquad t\in[0,T]\,.
\]
We recall that $[-h,0)$ is the pre-trading observation interval; trading takes place only on $[0,T]$.

Let the fixed Borel-measurable map  $\mathcal D:C([0,1];\mathbb R^d)\to\mathbb R^d$ be a  non-anticipative trend extractor for $(R(\omega),\Pi,h)$ for every $\omega\in\Omega$\,; then the construction of Subsection~\ref{subsec:nonanticipative-decomposition} yields the pathwise objects
\[
    A^{h,{\cal D}}(\omega),\qquad M^{h,{\cal D}}(\omega),\qquad C^{h,{\cal D}}(\omega),\qquad O^{h,{\cal D}}(\omega),\qquad  a^{h,{\cal D}}(\omega),\qquad c^{h,{\cal D}}(\omega),
\]
where
\[
    R(t,\omega)-R(0,\omega)=A^{h,{\cal D}}(t,\omega)+M^{h,{\cal D}}(t,\omega),\qquad t\in[0,T],
\]
and
\[
    C^{h,{\cal D}}(\omega)=\big[M^{h,{\cal D}}(\omega),M^{h,{\cal D}}(\omega)\big]_\Pi=\big[ R(\omega),R(\omega)\big]_\Pi.
\]
Thus, the same window length,   partition sequence, and trend extractor are used for every scenario. We write then $\mathfrak P^{h,{\cal D}}(\omega)$ and $\mathfrak W^{h,{\cal D}}(\omega)$ for the single-path portfolio classes constructed from the return path $R(\omega)$ in Definitions~\ref{def:portfolios} and \ref{def:follmer-integrable-portfolios}, respectively.

\begin{lemma}[Joint Measurability and Non-anticipativity]   \label{lem:joint-measurable-non-anticipative-characteristics}
    Equip $\Omega\subset C([-h,T];\mathbb R^d)$ with the subspace topology inherited from the topology of uniform convergence on compact intervals, and with the canonical $\sigma$-field generated by the coordinate maps. Then the maps
    \[
        (t,\omega)\longmapsto A^{h,\mathcal D}(t,\omega),\quad C^{h,\mathcal D}(t,\omega),\quad O^{h,\mathcal D}(t,\omega)
    \]
    and the non-anticipatively selected versions
    \[
        (t,\omega)\longmapsto a^{h,\mathcal D}(t,\omega),\quad c^{h,\mathcal D}(t,\omega)
    \]
    of Definition~\ref{def:pathwise-local-characteristics-extractor}, are jointly Borel measurable on $[0,T]\times\Omega$.

    These are also non-anticipative: if $\omega,\widetilde\omega\in\Omega$ agree on $[-h,t]$, then
    \[
        a^{h,\mathcal D}(t,\omega) = a^{h,\mathcal D}(t,\widetilde\omega), \qquad c^{h,\mathcal D}(t,\omega) = c^{h,\mathcal D}(t,\widetilde\omega);
    \]
    thus for every fixed $t$, the maps $\omega\mapsto a^{h,\mathcal D}(t,\omega)$ and $\omega\mapsto c^{h,\mathcal D}(t,\omega)$ are $\mathcal F^\Omega(t)$-measurable, and the same conclusions hold for
    \begin{equation}        \label{eq:3.9}
        \nu^{h,\mathcal D}(t,\omega) := \big(c^{h,\mathcal D}(t,\omega)\big)^\dagger a^{h,\mathcal D}(t,\omega).
    \end{equation}
\end{lemma}

\begin{proof}
    The map $(t,\omega)\mapsto R^{t,h}(\omega)$ from $[0,T]\times\Omega$ into $C([0,1];\mathbb R^d)$ is continuous under the topology of uniform convergence on compact intervals. Since $\mathcal D$ is Borel-measurable, the trend signal is jointly Borel-measurable; hence its time integral $A^{h,\mathcal D}$ is jointly Borel-measurable.

    For each $n$, the discrete quadratic-covariation sum along $\Pi_n$ is a jointly Borel-measurable function of $(t,\omega)$, and uses only values of $\omega$ up to time $t$. Its pointwise limit on $\Omega_\Pi$ is therefore jointly Borel-measurable and non-anticipative; this limit is $C^{h,\mathcal D}$.  Since
    \[
        \breve{A}_i^{h,\mathcal D}(t,\omega) = \int_0^t |\alpha_i^{h,\mathcal D}(s,\omega)|\,\d s,
    \]
    the clock $O^{h,\mathcal D}$ has the same properties. Each rational one-sided quotient in Definition~\ref{def:non-anticipative-one-sided-density} is jointly Borel-measurable and uses only the history up to $t$; taking the rational limit, and using the fixed zero convention, preserves both properties.  Finally, the pseudo-inverse is a Borel-measurable map on the finite-dimensional space of symmetric matrices, which proves the assertion for $\nu^{h,\mathcal D}$.
\end{proof}

\begin{definition}[Scenario-wise Wealth-generating Portfolios]  \label{def:scenario-wise-follmer-integrable}
    A map $\pi:[0,T]\times\Omega\to\mathbb R^d$ is a \emph{scenario-wise wealth-generating F\"ollmer portfolio}, if
    \begin{enumerate}[label=(\roman*)]
        \item $(t,\omega)\mapsto\pi(t,\omega)$ is jointly Borel-measurable;
        \item for every $t$, the map $\omega\mapsto\pi(t,\omega)$ is $\mathcal F^\Omega(t)$-measurable;
        \item for every $\omega\in\Omega$, the path $\pi(\cdot,\omega)$ belongs to the concrete class $\mathfrak W^{h,\mathcal D}(\omega)$ of Definition~\ref{def:follmer-integrable-portfolios} constructed from $R(\cdot,\omega)$.
    \end{enumerate}
    The collection of all such rules is denoted by $\mathfrak W_\Omega^{h,\mathcal D}$.
\end{definition}

For $\pi,\rho\in\mathfrak W_\Omega^{h,\mathcal D}$ and
$u\in\mathbb D_\Pi$, define their concatenation at time $u$ by
\begin{equation*}
    (\pi\otimes_u\rho)(t,\omega) := \pi(t,\omega)\mathbf 1_{[0,u)}(t) + \rho(t,\omega)\mathbf 1_{[u,T]}(t), \qquad (t,\omega)\in[0,T]\times\Omega.
\end{equation*}
The convention at $u$ assigns the increment immediately after the
switching time to the new portfolio $\rho$ in the left-point Riemann sums.

\begin{corollary}[Pasting at Partition Times]   \label{cor:pasting-partition-times}
    Let $\pi,\rho\in\mathfrak W_\Omega^{h,\mathcal D}$ and $u\in\mathbb D_\Pi$.  Then $\pi\otimes_u\rho$ belongs to $\mathfrak W_\Omega^{h,\mathcal D}$; and for every $t$ and $\omega$,
    \begin{align*}
        R_{\pi\otimes_u\rho}(t,\omega) &=R_\pi(t\wedge u,\omega) +R_\rho(t,\omega)-R_\rho(t\wedge u,\omega),
        \\
        M_{\pi\otimes_u\rho}^{h,\mathcal D}(t,\omega) & = M_\pi^{h,\mathcal D}(t\wedge u,\omega)+M_\rho^{h,\mathcal D}(t,\omega)-M_\rho^{h,\mathcal D}(t\wedge u,\omega),
        \\
        X_{\pi\otimes_u\rho}(t,\omega) &=X_\pi(t\wedge u,\omega) \frac{X_\rho(t,\omega)}{X_\rho(t\wedge u,\omega)}.
    \end{align*}
    Similar splitting identities hold for $A$, $C$, and all cross covariations.
\end{corollary}

\begin{proof}
    Measurability and non-anticipativity are immediate.  Pathwise membership follows directly from the finite $\mathbb D_\Pi$-pasting closure in Definition~\ref{def:follmer-integrable-portfolios}. For all sufficiently large $n$, the refining partition $\Pi_n$ contains $u$, so every defining Riemann sum splits exactly at $u$.  Passing to the uniform limits gives the integral identities; the characteristic and covariation identities follow by splitting their Lebesgue--Stieltjes integrals; and the wealth identity follows by exponentiation.
\end{proof}

\medskip

\subsection{Pathwise viability} \label{subsec:pathwise-viability}

We introduce now cumulative withdrawal streams and their pathwise financing capital. A cumulative withdrawal stream is required to be non-anticipative and nondecreasing in time. We allow RCLL (i.e., Right-Continuous, with Left-Limits) withdrawal streams, so that lump-sum payments, including payments at the terminal time, are included in the framework.

\begin{definition}[Cumulative Capital Withdrawal Stream]    \label{def:cumulative-withdrawal-stream}
    A mapping $K:[0,T]\times\Omega\to[0,\infty)$ is called  \emph{cumulative withdrawal stream,} if the following hold:
    \begin{enumerate}[label=(\roman*)]
        \item $K (0, \omega)=0$ for every scenario $\omega\in\Omega$.
        \item For every scenario $\omega\in\Omega$, the path $t\mapsto K (t, \omega)$ is RCLL and nondecreasing on $[0,T]$.
        \item For every $t\in[0,T]$, the map $\omega\mapsto K (t, \omega)$ is $\mathcal F^\Omega (t)$-measurable.
    \end{enumerate}
    We denote by $\mathcal K_\Omega$ the collection of all cumulative capital withdrawal streams.
\end{definition}

Let $\pi\in\mathfrak W_\Omega^{h,\mathcal D}$ and $x\ge0$. The wealth generated by initial capital $x$ without withdrawals, is $xX_\pi$. If a cumulative withdrawal stream $K$ is paid out, while the remaining wealth continues to be invested according to the portfolio $\pi$, the  resulting wealth after withdrawals is given by
\begin{equation}  \label{eq:pathwise-wealth-withdrawal}
    V_{x,\pi,K}(t,\omega) := X_\pi(t,\omega) \Bigg( x-\int_{(0,t]}\frac{1}{X_\pi(s,\omega)}\,\d K (s, \omega) \Bigg), \qquad t\in[0,T].
\end{equation}
The integral in \eqref{eq:pathwise-wealth-withdrawal} is the pathwise Lebesgue--Stieltjes integral with respect to the finite positive measure induced by the nondecreasing RCLL  path $K(\cdot,\omega)$. It is well-defined since $X_\pi(\cdot,\omega)$ is continuous and strictly positive.

The formula \eqref{eq:pathwise-wealth-withdrawal} gives a pathwise solution of the self-financing wealth equation with cumulative withdrawals. Indeed, suppressing $\omega$ from the notation momentarily, and setting
\[
    H(t):=x-\int_{(0,t]}\frac{1}{X_\pi (s)}\, \d K(s),
\]
we have $V_{x,\pi,K}=X_\pi H$. The path $H$ is RCLL and of finite variation, with
\[
    \d H (t)=-\frac{1}{X_\pi (t)}\, \d K (t)\,.
\]
Since $X_\pi$ is continuous, the pathwise integration-by-parts formula gives
\begin{equation*}
    \d V_{x,\pi,K} (t) = H (t-)\,\d X_\pi (t)+X_\pi (t)\,\d H (t) = \frac{V_{x,\pi,K}(t-)}{X_\pi (t)}\,\d X_\pi (t)-\d K (t), \quad V_{x,\pi,K}(0)=x.
\end{equation*}
In particular, $\,\Delta V_{x,\pi,K} (t)=-\Delta K (t)\,$ holds, so a lump-sum withdrawal produces an equal downward jump in the remaining wealth.

\begin{definition}[Scenario-wise Financeability]    \label{def:scenario-wise-financeability}
    Let $x\ge0$ and $K\in\mathcal K_\Omega$.     We say that the capital withdrawal stream $K$  can be \emph{financed starting with  initial capital $x$\,,} if there exists a portfolio $\pi\in\mathfrak W_\Omega^{\,h,{\cal D}}$ with
    \[
        V_{x,\pi,K}(t,\omega)\ge0, \qquad \forall\,t\in[0,T],\quad \forall\,\omega\in\Omega.
    \]
\end{definition}
    
\medskip
In light of \eqref{eq:pathwise-wealth-withdrawal}, this requirement can be cast equivalently as 
\[
    \int_{(0,t]}\frac{1}{X_\pi(s,\omega)}\,\d K (s, \omega)\le x, \qquad \forall\,t\in[0,T],\quad \forall\,\omega\in\Omega
\]
and, since the left-hand side is nondecreasing in $t$, is equivalent to
\begin{equation*}
    \int_{(0,T]}\frac{1}{X_\pi(s,\omega)}\,\d K (s, \omega)\le x, \qquad \forall\,\omega\in\Omega.
\end{equation*} 

\begin{definition}[Scenario-wise Financing Capital]    \label{def:scenario-wise-financing-capital}
    For $x\ge0$, we define now the collection
    \begin{equation}  \label{eq:3.5}
        \mathcal K_\Omega(x) := \big\{ K\in\mathcal K_\Omega: \text{$K$ can be financed from initial capital $x$} \big\}
    \end{equation} 
    and, in terms of it, the {\it scenario-wise financing capital of $K\in\mathcal K_\Omega$} as
    \begin{equation} \label{eq:3.6}
        x_\Omega(K) := \inf\big\{ x\ge0: K\in\mathcal K_\Omega(x) \big\}.
    \end{equation} 
\end{definition}
   
In words, $\mathcal K_\Omega(x)$ is the collection of future capital withdrawal streams that can be financed starting with initial capital $x\ge 0\,;$ and $ x_\Omega(K)$ the infimal initial capital starting with which a given future capital withdrawal stream $K \in\mathcal K_\Omega$ can be financed.
Now, a portfolio $\pi\in\mathfrak W_\Omega^{h,{\cal D}}$ can finance a given $K\in\mathcal K_\Omega$ starting   from initial capital $x$ if, and only if, 
\[
    \sup_{\omega\in\Omega} \int_{(0,T]} \frac{1}{X_\pi(s,\omega)}\,\d K (s, \omega) \le x.
\]
Consequently, and with the convention that the infimum of the empty set is $ + \infty\,,$ the   scenario-wise financing capital of $K\in\mathcal K_\Omega$ in \eqref{eq:3.6} is given by 
\[
    x_\Omega(K) = \inf_{\pi\in\mathfrak W_\Omega^{h,\mathcal D}} \sup_{\omega\in\Omega} \,\int_{(0,T]} \frac{1}{X_\pi(s,\omega)}\,\d K(s, \omega).
\]

The following definition is the robust pathwise counterpart of the viability principle articulated in \cite{KaratzasKardaras2021}; this posits that \emph{it should not be possible to finance a nontrivial cumulative withdrawal stream, starting with arbitrarily small initial capital}.

\begin{definition}[Scenario-wise Pathwise Viability]    \label{def:scenario-wise-pathwise-viability}
    The scenario market $(\Omega,\mathfrak W_\Omega^{\,h,{\cal D}})$ is called \emph{pathwise viable} on $[0,T]$ if, for every $K\in\mathcal K_\Omega\,$,
    \[
        x_\Omega(K)=0 \qquad\text{implies} \qquad K\equiv0 \quad\text{on }[0,T]\times\Omega.
    \]
    Equivalently, and in the notation of \eqref{eq:3.5}, if
    \[
        \mathcal K_\Omega(0+) := \bigcap_{\varepsilon>0}\mathcal K_\Omega(\varepsilon)=\{0\}.
    \]
\end{definition}

\subsection{Pointwise boundedness of wealth and pathwise viability}\label{subsec:pointwise-boundedness-viability}

We characterize now pathwise viability in terms of {\it pointwise boundedness of attainable wealth}. This notion is a strengthening of the ``boundedness-in-probability"  in Proposition 2.22 of \cite{KaratzasKardaras2021}.

We start by noting that the zero portfolio $\,\pi_0 \equiv 0\,$ belongs to $\mathfrak W_\Omega^{\,h,{\cal D}}$. Indeed, it is Borel measurable and non-anticipative, and its restriction to every scenario belongs to the corresponding linear single-path class $\mathfrak W^{\,h,{\cal D}}(\omega)$. Moreover,
\[
    R_{\pi_0}\equiv0,\qquad M_{\pi_0}\equiv0,\qquad X_{\pi_0}\equiv1.
\]
For $\pi\in\mathfrak W_\Omega^{\,h,{\cal D}}$ and $u\in\mathbb D_\Pi$, define the portfolio stopped at $u$ by
\[
    \pi^u:=\pi\otimes_u \pi_0\,.
\]
Corollary~\ref{cor:pasting-partition-times} gives that, for $\,t\in[0,T],~ \omega\in\Omega\,,$ this satisfies
\[
    R_{\pi^u}(t,\omega)=R_\pi(t\wedge u,\omega), \quad M_{\pi^u}^{h,\mathcal D}(t,\omega)=M_\pi^{h,\mathcal D}(t\wedge u,\omega),\quad X_{\pi^u}(t,\omega)=X_\pi(t\wedge u,\omega),
\] 
Thus, if trading stops at a partition time,  the wealth attained at that time is preserved.

\begin{definition}[Scenario-wise Pointwise Boundedness] \label{def:scenario-wise-pointwise-boundedness}
    The collection of attainable levels  of terminal wealth   is called \emph{scenario-wise pointwise bounded,} if
    \begin{equation} \label{eq:3.7}
        \mathcal B (T, \omega) := \sup_{\pi\in\mathfrak W_\Omega^{\,h,{\cal D}}}X_\pi(T,\omega)<\infty, \qquad \forall\,\omega\in\Omega.
    \end{equation}
\end{definition}

The next lemma shows that terminal pointwise boundedness also controls wealth attained before the terminal time.

\begin{lemma}[Terminal and Running Pointwise Boundedness]\label{lem:terminal-running-pointwise-boundedness}
    For every $\omega\in\Omega$,
    \[
        \sup_{\pi\in\mathfrak W_\Omega^{h,\mathcal D}}X_\pi(T,\omega) = \sup_{\pi\in\mathfrak W_\Omega^{h,\mathcal D}} \sup_{t\in[0,T]}X_\pi(t,\omega).
    \]
    Consequently, the terminal wealth family is scenario-wise pointwise bounded if, and only if,
    \[
        \sup_{\pi\in\mathfrak W_\Omega^{h,\mathcal D}} \sup_{t\in[0,T]}X_\pi(t,\omega) <\infty, \qquad \forall\,\omega\in\Omega.
    \]
\end{lemma}

\begin{proof}
    The left-hand side is bounded above by the right-hand side. Conversely, fix $\pi\in\mathfrak W_\Omega^{h,\mathcal D}$, $t\in[0,T]$, and $\omega\in\Omega$. Since $\mathbb D_\Pi$ is dense in $[0,T]$, there exists a sequence $(u_m)_{m\in\mathbb N}\subset\mathbb D_\Pi$ such that $u_m\to t$. By continuity of $X_\pi(\cdot,\omega)$,
    \[
        X_\pi(t,\omega) = \lim_{m\to\infty}X_\pi(u_m,\omega).
    \]
    But for each $m \in \N$, the stopped portfolio $\pi^{u_m}$ belongs to $\mathfrak W_\Omega^{h,\mathcal D}$ and satisfies $X_{\pi^{u_m}}(T,\omega)=X_\pi(u_m,\omega)$, therefore 
    \[
        X_\pi(t,\omega) \le \sup_{\rho\in\mathfrak W_\Omega^{h,\mathcal D}}X_\rho(T,\omega).
    \]
    Taking the supremum over $\pi$ and $t$ proves the reverse inequality.
\end{proof}

Here is the main result of this subsection; it is the analogue of Proposition 2.22 in \cite{KaratzasKardaras2021}.

\begin{theorem}[Pointwise Boundedness and Pathwise Viability]   \label{thm:pointwise-boundedness-viability}
    The following  are equivalent:
    \begin{enumerate}[label=(\roman*)]
        \item The scenario market $(\Omega,\mathfrak W_\Omega^{\,h,{\cal D}})$ is pathwise viable on $[0,T]$.
        \item The terminal wealth family is scenario-wise pointwise bounded, i.e., \eqref{eq:3.7} holds.
    \end{enumerate}
\end{theorem}

\begin{proof}
    Assume first that the terminal wealth family is scenario-wise pointwise bounded. Let $K\in\mathcal K_\Omega$ be nonzero; there exists then $\omega_0\in\Omega$ such that
$\,
        K (T, \omega_0)>0.
    \,$ 
    For every $\pi\in\mathfrak W_\Omega^{h,\mathcal D}$, Lemma~\ref{lem:terminal-running-pointwise-boundedness} gives
    \[
        \sup_{t\in[0,T]}X_\pi(t,\omega_0) \le \mathcal B (T, \omega_0)\,,
    \]
thus
    \[
        \int_{(0,T]} \frac{1}{X_\pi(s,\omega_0)} \, \d K(s, \omega_0) \ge \frac{K (T, \omega_0)}{\sup_{t\in[0,T]}X_\pi(t,\omega_0)} \ge \frac{K (T, \omega_0)}{\mathcal B (T, \omega_0)}
    \]
and
    \[
        x_\Omega(K) = \inf_{\pi\in\mathfrak W_\Omega^{\,h,{\cal D}}} \sup_{\omega\in\Omega} \,\int_{(0,T]} \frac{1}{X_\pi(s,\omega)}\,\d K (s, \omega) \ge \frac{K (T, \omega_0)}{\mathcal B (T, \omega_0)} > 0.
    \]
    Therefore $x_\Omega(K)=0$ implies $K\equiv0$, and the scenario market is pathwise viable.

    Conversely, suppose that scenario-wise pointwise boundedness fails. Then there exists $\omega_0\in\Omega$ such that
    \[
        \sup_{\pi\in\mathfrak W_\Omega^{\,h,{\cal D}}}X_\pi(T,\omega_0)=\infty.
    \]
    For every $n\in\mathbb N$, choose $\pi^n\in\mathfrak W_\Omega^{\,h,{\cal D}}$ such that
  $\,
        X_{\pi^n}(T,\omega_0)\ge n\,,
    \,$
    and define the terminal lump-sum withdrawal stream
    \[
        K^{\omega_0} (t, \omega) := \mathbf 1_{\{T\}}(t)\,\, \mathbf 1_{\{\omega_0\}}(\omega)\,.
    \]
    This stream is nonnegative, nondecreasing, and RCLL. Moreover, $\{\omega_0\}\in\mathcal F^\Omega (T)$: since paths in $\Omega$ are continuous,
    \[
        \{\omega_0\} = \bigcap_{q\in\mathbb Q\cap[-h,T]} \Big\{ \omega\in\Omega: R(q,\omega)=R(q,\omega_0) \Big\},
    \]
thus $K^{\omega_0}\in\mathcal K_\Omega$. We set now
    \[
        x_n:=\frac{1}{X_{\pi^n}(T,\omega_0)}\le\frac1n.
    \]
    Before time $T$, no withdrawal is made; at time $T$, we have 
    \[
        V_{x_n,\pi^n,K^{\omega_0}}(T,\omega_0) = x_n\, X_{\pi^n}(T,\omega_0)-1 = 0\,;
    \]
    while, for $\omega\ne\omega_0$,
    \[
        V_{x_n,\pi^n,K^{\omega_0}}(T,\omega) = x_n\, X_{\pi^n}(T,\omega)>0.
    \]
    Thus
    $\,
        K^{\omega_0}\in\mathcal K_\Omega(x_n), ~~ \forall\,n\in\mathbb N
    \,,$
    and consequently 
    $\,
        x_\Omega\big(K^{\omega_0}\big)=0.
    \,$ 
    Since $K^{\omega_0}$ is nonzero, pathwise viability fails.
\end{proof}

\begin{remark}[Pathwise Analogue of Boundedness-in-Probability] \label{rem:pathwise-boundedness-analogue}
    Scenario-wise pointwise boundedness is stronger than boundedness in probability~(e.g., \cite[Proposition~2.22]{KaratzasKardaras2021}): it requires boundedness separately on every scenario, and does not discard an exceptional set. 
    
    This strengthening is natural in the present robust pathwise setting, where no probability measure is specified and viability is required uniformly over $\Omega$. The reverse implication in Theorem~\ref{thm:pointwise-boundedness-viability} relies on the availability of scenario-specific terminal lump-sum withdrawal streams.
\end{remark}

\subsection{A pathwise finite-resolution cornerstone theorem}   \label{subsec:pathwise-finite-resolution-cornerstone}

We summarize the two pathwise equivalence results established above. The first concerns the growth structure and the existence of a driftless num\'eraire portfolio. The second concerns scenario-wise financing, and identifies pathwise viability with pointwise boundedness of attainable terminal wealth levels. It is a very distinct feature of the theory developed here, that these two equivalence classes are logically distinct in the general pathwise setting.

\begin{definition}[Canonical Num\'eraire Admissibility]   \label{def:canonical-numeraire-admissibility}
    For the non-anticipatively selected scenario-wise local characteristics $a^{h,\mathcal D}$ and $c^{h,\mathcal D}$ of Definition~\ref{def:pathwise-local-characteristics-extractor}, we say that \emph{canonical num\'eraire admissibility} holds on $\Omega$ if, whenever
    \begin{equation}    \label{eq:3.8}
        G^{h,{\cal D}}(T,\omega)<\infty,\qquad \forall\,\omega\in\Omega,
    \end{equation} 
    the canonical portfolio $\nu^{h,{\cal D}} (t, \omega) := \big(c^{h,{\cal D}}(t, \omega)\big)^\dagger a^{h,{\cal D}}(t, \omega)$ of \eqref{eq:3.9} belongs to the collection $\mathfrak W_\Omega^{\,h,{\cal D}}$ of Definition~\ref{def:scenario-wise-follmer-integrable}.
\end{definition}

A direct sufficient condition for canonical num\'eraire admissibility is available from the definition of the basic admissible F\"ollmer integrands. Namely, suppose that, whenever \eqref{eq:3.8} holds, there exist an integer $m\ge1$, a jointly Borel-measurable and non-anticipative map
\[
    B:[0,T]\times\Omega\longrightarrow\mathbb R^m,
\]
whose path $B(\cdot,\omega)$ is continuous and of finite first variation for every $\omega\in\Omega$, and a function $f\in C^{1,2}(\mathbb R^m\times\mathbb R^d)$ such that
    \begin{equation}    \label{eq:3.8.a}
    \nu^{h,\mathcal D}(t,\omega) = \nabla_x f\big(B(t,\omega),R(t,\omega)\big), \qquad (t,\omega)\in[0,T]\times\Omega.
    \end{equation} 
Then, for every $\omega\in\Omega$, the path $\nu^{h,\mathcal D}(\cdot,\omega)$ is a basic admissible F\"ollmer integrand for $R(\cdot,\omega)$ as in Definition~\ref{def:basic-admissible-follmer-integrands}. Joint Borel measurability and non-anticipativity therefore imply $\nu^{h,\mathcal D}\in\mathfrak W_\Omega^{h,\mathcal D}$, and canonical num\'eraire admissibility holds on $\Omega$.

The preceding condition is expressed directly in the admissible-integrand form. We record next a more concrete, although stronger, sufficient condition which can often be checked from the time regularity of the canonical rule.

\begin{proposition}[A Finite Variation Criterion]   \label{prop:finite-variation-canonical-admissibility}
    Assume that \eqref{eq:3.8} holds; that $\nu^{h,\mathcal D}$ is jointly Borel-measurable and non-anticipative; and that, after changing its values at finitely-many deterministic times in $\mathbb D_\Pi$ if necessary, each path $\nu^{h,\mathcal D}(\cdot,\omega)$ is a finite $\mathbb D_\Pi$-pasting of continuous finite-variation paths. Then $\nu^{h,\mathcal D} \in\mathfrak W_\Omega^{h,\mathcal D}$, hence canonical num\'eraire admissibility holds on $\Omega$.
\end{proposition}

\begin{proof}
    Fix $\omega\in\Omega$ and let $b^\ell(\cdot,\omega)$ be the continuous finite-variation path representing the $\ell$th piece of the finite pasting. On this piece, read         \eqref{eq:3.8.a} with 
    \[
        B=b^\ell(\cdot,\omega), \qquad f(b,x)=b^\top x.
    \]
    Then $\nabla_x f(b,x)=b$, so each piece is a basic admissible F\"ollmer integrand. The conclusion follows from the finite $\mathbb D_\Pi$-pasting and finite-point modification closure.
\end{proof}

\begin{corollary}[A Calendar-Time Criterion]    \label{cor:calendar-time-canonical-admissibility}
    Assume that \eqref{eq:3.8} holds; and that for every $\omega\in\Omega$ there is a continuous,  symmetric, nonnegative-definite matrix path $\Sigma^{h,\mathcal D}(\cdot,\omega)$ with
    \[
        C^{h,\mathcal D}(t,\omega) = \int_0^t \Sigma^{h,\mathcal D}(s,\omega)\,\d s,
    \]
    and write $\alpha^{h,\mathcal D}$ for the calendar-time trend signal in \eqref{def:trend-signal}. Suppose that the maps $\alpha^{h,\mathcal D}$ and $\Sigma^{h,\mathcal D}$ are jointly Borel-measurable and non-anticipative; that, for every $\omega\in\Omega$, both paths
    \[
        \alpha^{h,\mathcal D}(\cdot,\omega) \quad\text{and}\quad \Sigma^{h,\mathcal D}(\cdot,\omega)
    \]
    are continuous and of finite first variation; and that, for every $\omega\in\Omega$, the path $\Sigma^{h,\mathcal D}(\cdot,\omega)$ has constant rank. Assume also that
    \[
       {\bm o} (t,\omega):=\sum_{i=1}^d|\alpha_i^{h,\mathcal D}(t,\omega)|
        +\operatorname{tr}\Sigma^{h,\mathcal D}(t,\omega) \ge \varepsilon_\omega>0, \qquad 0\le t\le T.
    \]
    
    Then, for $0<t\le T$, the non-anticipatively selected characteristics are
    \[
        a^{h,\mathcal D}(t,\omega) = \frac{\alpha^{h,\mathcal D}(t,\omega)}{{\bm o} (t,\omega)}, \qquad c^{h,\mathcal D}(t,\omega) = \frac{\Sigma^{h,\mathcal D}(t,\omega)}{{\bm o} (t,\omega)},
    \]
    and the canonical portfolio has the simpler representation
    \[
        \nu^{h,\mathcal D}(t,\omega) = \big(\Sigma^{h,\mathcal D}(t,\omega)\big)^\dagger \alpha^{h,\mathcal D}(t,\omega),
        \qquad 0<t\le T.
    \]
    At $t=0$ the non-anticipative convention of Definition~\ref{def:non-anticipative-one-sided-density} assigns the value zero. Consequently, canonical num\'eraire admissibility holds whenever the locally finite growth condition is satisfied.
\end{corollary}

\begin{proof}
    With the present choices, the clock in \eqref{def:O-clock} is absolutely continuous with density ${\bm o}$. The one-sided selection therefore agrees on $(0,T]$ with the displayed continuous ratios. For ${\bm o} >0$, the scaling identity $(\Sigma/{\bm o} )^\dagger={\bm o} \Sigma^\dagger$ gives the formula for $\nu$. Continuity, constant rank, and compactness of $[0,T]$ imply that the positive spectrum of $\,\Sigma^{h,\mathcal D}(\cdot,\omega)\,$ is uniformly bounded away from zero. Thus, the pseudo-inverse is Lipschitz on the compact range of this path, and $\big(\Sigma^{h,\mathcal D}\big)^\dagger$ has finite variation. Hence, Proposition~\ref{prop:finite-variation-canonical-admissibility} applies.
\end{proof}

We recall the collection $\mathfrak W_\Omega^{h,\mathcal D}$ from Definition~\ref{def:scenario-wise-follmer-integrable}, and are now ready to state the pathwise cornerstone theorem.

\begin{theorem}[Pathwise Cornerstone Theorem] \label{thm:pathwise-finite-resolution-cornerstone}
    Suppose that canonical num\'eraire admissibility holds on the set of scenarios $\Omega$.
    
    \medskip
    
    \noindent\textnormal{\bf (A) Growth--Num\'eraire layer.}
    The following conditions are equivalent:
    \begin{enumerate}[label=(\roman*)]
        \item Every scenario has locally finite growth, i.e., \eqref{eq:3.8} holds.

        \item For every $\omega\in\Omega$, we have 
        \[
            a^{h,{\cal D}}(t, \omega) \in \operatorname{Range}\big(c^{h,{\cal D}}(t, \omega)\big), \qquad \d O^{h,{\cal D}}(\omega)\text{-a.e. on }[0,T],
        \]
        \[
            \int_0^T a^{h,{\cal D}}(t,\omega)^\top \big(c^{h,{\cal D}}(t,\omega)\big)^\dagger a^{h,{\cal D}}(t,\omega) \,\d O^{h,{\cal D}}(t,\omega) <\infty.
        \]

        \item There exists $\,\nu\in\mathfrak W_\Omega^{h,{\cal D}}\,$ such that $\nu(\cdot,\omega)$ is a driftless num\'eraire portfolio for every $\omega\in\Omega$.
    
        \item There exists $\nu\in\mathfrak W_\Omega^{h,{\cal D}}$ such that $\nu(\cdot,\omega)$ is a growth-optimal portfolio for every $\omega\in\Omega$.

        \item There exists $\nu\in\mathfrak W_\Omega^{\,h,{\cal D}}$ such that, for every $\omega\in\Omega$ and every $p\in\mathfrak P^{h,{\cal D}}(\omega)$,
        \[
            \Gamma_\nu^{h,{\cal D}}(t,\omega) - \Gamma_p^{h,{\cal D}}(t,\omega) = \frac12 \,C_{p-\nu(\cdot,\omega)}^{h,{\cal D}}(t,\omega), \qquad t\in[0,T].
        \]
    \end{enumerate}
    
    \noindent\textnormal{\bf (B) Viability--Boundedness layer.}
    The following conditions are equivalent:
    \begin{enumerate}[label=(\roman*),resume]
        \item The scenario market $(\Omega,\mathfrak W_\Omega^{h,{\cal D}})$ is pathwise viable on $[0,T]$.
    
        \item The family of attainable terminal wealth levels, is scenario-wise pointwise bounded:
        \[
            \mathcal B(T, \omega) := \sup_{\pi\in\mathfrak W_\Omega^{h,{\cal D}}} X_\pi(T,\omega) <\infty, \qquad \forall\,\omega\in\Omega.
        \]
    \end{enumerate}
\end{theorem}

\begin{proof}
    The equivalence of conditions {\rm(i)}--{\rm(v)} follows from Proposition~\ref{prop:finite-resolution-growth-structural-benchmark}, applied separately to each scenario $\omega\in\Omega$. Canonical num\'eraire admissibility ensures that the canonical rule $\nu$ of \eqref{eq:3.9} 
    belongs to $\,\mathfrak W_\Omega^{\,h,{\cal D}}$ whenever the equivalent growth conditions hold. Proposition~\ref{prop:driftless-numeraire}, applied path by path, identifies each restriction $\nu^{h,{\cal D}}(\cdot,\omega)$ as a driftless num\'eraire portfolio. The equivalence of conditions {\rm(vi)} and {\rm(vii)} is Theorem~\ref{thm:pointwise-boundedness-viability}.
\end{proof}

The two layers of Theorem~\ref{thm:pathwise-finite-resolution-cornerstone} cannot, in general, be combined into a single equivalence class. In particular, \emph{local finiteness of growth does not by itself imply scenario-wise pointwise boundedness or pathwise viability}. The following two examples show that this failure may occur even when canonical num\'eraire admissibility holds.

\begin{example} [An $N$-Resolution Single-Scenario Separation] \label{ex:lfg-without-pointwise-boundedness}
    Consider a one-dimensional market and use the last-cell Faber--Schauder trend extractor from Example~\ref{ex:last-FS-slope}. Fix $h>0$ and $N\in\mathbb N$, and set $\delta:=h2^{-N}$. Throughout this example, we write $(h,N)$ as shorthand for $(h,\mathcal D_N^{\mathrm{last}})$. Let $\omega\in C([-h,T];\mathbb R)$ be the smooth path
    \[
        \omega(t):=\sin\Big(\frac{2\pi t}{\delta}\Big), \qquad t\in[-h,T],
    \]
    and set $\Omega:=\{\omega\}$. Since $\omega$ is continuously differentiable, it belongs to $\Omega_\Pi$ in \eqref{def:Omega-Pi}. By the coordinate convention of Subsection~\ref{subsec:scenario-sets}, $R(t,\omega)=\omega(t)$ for $t\in[-h,T]$.
    
    For every $t\in[0,T]$, the $N$-resolution Faber--Schauder projection $P_NR^{t,h}$ agrees with the past-window path $R^{t,h}$ at the level-$N$ dyadic grid points. In particular,
    \[
        P_NR^{t,h}(1)=R(t,\omega),\qquad P_NR^{t,h}(1-2^{-N})=R(t-\delta,\omega).
    \]
    Since $R$ is $\delta$-periodic, $R(t,\omega)=R(t-\delta,\omega)$, and hence the last-cell trend signal vanishes:
    \[
        \alpha_{\mathrm{last}}^{h,N}(t,\omega)=\frac{P_NR^{t,h}(1)-P_NR^{t,h}(1-2^{-N})}{h2^{-N}}=0,\qquad t\in[0,T].
    \]
    Consequently,
    \[
        A^{h,N}(t,\omega)\equiv0,\qquad M^{h,N}(t,\omega)=R(t,\omega)-R(0,\omega)=R(t,\omega).
    \]
    
    The path $R(\cdot,\omega)$ is continuously differentiable, therefore has zero quadratic variation along $\Pi$. Thus
    \[
        C^{h,N}(\omega)=[R(\omega),R(\omega)]_\Pi\equiv0.
    \]
    Since both $A^{h,N}(\omega)$ and $C^{h,N}(\omega)$ vanish identically, the operational clock satisfies $O^{h,N}(\omega)\equiv0$, and the non-anticipatively selected characteristics are identically zero:
    \[
        a^{h,N}(t, \omega)\equiv0,\qquad c^{h,N}(t, \omega)\equiv0.
    \]
    It follows that $g^{h,N}(t,\omega)=0$ and $G^{h,N}(T,\omega)=0$. Hence the return path has $N$-resolution locally finite growth. The canonical portfolio is
    \[
        \nu^{h,N}(t, \omega)=\big(c^{h,N}(t, \omega)\big)^\dagger a^{h,N}(t, \omega)\equiv0.
    \]
    Since the zero portfolio belongs to $\mathfrak W_\Omega^{h,N}$, canonical num\'eraire admissibility also holds in this example.
    
    For $m\in\mathbb N$, consider the continuously differentiable portfolio
    \[
        \pi^m(t,\omega):=m\dot R(t,\omega), \qquad t\in[0,T].
    \]
    The path $\pi^m(\cdot,\omega)$ is continuous and of finite variation on $[0,T]$. By Definitions~\ref{def:basic-admissible-follmer-integrands} and~\ref{def:scenario-wise-follmer-integrable}, the rule $\pi^m$ belongs to $\mathfrak W_\Omega^{h,N}$: take $B(t):=\pi^m(t,\omega)$, $f(b,x):=bx$ so that
    \[
        \partial_x f\big(B(t),R(t,\omega)\big)=B(t)=\pi^m(t,\omega).
    \]
    Since $R(\cdot,\omega)$ is continuously differentiable, the corresponding F\"ollmer integral agrees with the classical Riemann--Stieltjes integral:
    \[
        R_{\pi^m}(T,\omega)=m\int_0^T\dot R(t,\omega)\,\d R(t,\omega)=m\int_0^T\big|\dot R(t,\omega)\big|^2\,\d t.
    \]
    Moreover, $C_{\pi^m}^{h,N}(T,\omega)=0$. Therefore
    \[
        X_{\pi^m}(T,\omega) = \exp\bigg(m\int_0^T\big|\dot R(t,\omega)\big|^2\, \d t\bigg) \longrightarrow  \infty  
    \]
    as $m \to \infty$. It follows that
    \[
        \mathcal B (T, \omega)=\sup_{\pi\in\mathfrak W_\Omega^{h,N}}X_\pi(T,\omega)=\infty.
    \]
    Thus the attainable terminal wealth family is not scenario-wise pointwise bounded. By Theorem~\ref{thm:pointwise-boundedness-viability}, the scenario market is not pathwise viable.
    
    This failure can also be seen directly. Define the terminal lump-sum withdrawal stream by
    \[
        K(t, \omega):=\mathbf 1_{\{ T\}} (t)\,,\qquad t\in[0,T].
    \]
    This stream is nonnegative, nondecreasing, RCLL, and nonzero. For each $m\in\mathbb N$, set
    \[
        x_m:=\frac{1}{X_{\pi^m}(T,\omega)}.
    \]
    No withdrawal is made before time $T$, while at the terminal time,
    \[
        V_{x_m,\pi^m,K}(T,\omega)=X_{\pi^m}(T,\omega)\Bigg(x_m-\frac{1}{X_{\pi^m}(T,\omega)}\Bigg)=0.
    \]
    Hence $K\in\mathcal K_\Omega(x_m)$ for every $m$. Since $x_m\to0$, this gives $\,x_\Omega(K)=0\,,$ although $K\not\equiv0$; and confirms directly that pathwise viability fails.
\end{example}

The preceding example establishes the separation for a concrete last-cell Faber--Schauder extractor in the simplest singleton setting. In a singleton scenario set, however, non-anticipativity imposes no restriction across alternative paths. The next example complements it by considering an uncountable family of scenarios that coincide up to a common branching time, and by using a single sequence of jointly Borel-measurable, non-anticipative portfolio rules, each applied consistently across all scenarios. It also strengthens the pointwise wealth explosion to uniform explosion over the entire scenario set.

\begin{example}[A Multi-scenario Separation] \label{ex:multi-scenario-non-anticipative-separation}
    Assume $0<h<T$, fix $\tau \in (T-h, T)$ and $\underline\theta\in(0,1)$, and consider the nontrivial continuous ``lag-point'' trend extractor
    \[
        \mathcal D_{\mathrm{lag}}(x):=x(0), \qquad x\in C([0,1];\mathbb R).
    \]
    Let $\Theta_*:=[-1,-\underline\theta] \cup[\underline\theta,1]$. For each $\theta\in\Theta_*$, define the path
    \[
        \omega_\theta(t):=\theta\big((t-\tau)^+\big)^3, \qquad t\in[-h,T],
    \]
    and set $\Omega := \{\omega_\theta : \theta\in\Theta_*\}$. Every $\omega_\theta$ is continuously differentiable and hence belongs to $\Omega_\Pi$. By the coordinate convention of Subsection~\ref{subsec:scenario-sets},
    \begin{equation*}
        R(t,\omega_\theta) = \omega_\theta(t) = \theta\big((t-\tau)^+\big)^3, \qquad (t,\theta)\in[-h,T]\times\Theta_*.
    \end{equation*}
    Thus, an uncountable collection of scenarios has exactly the same history up to and including time $\tau$, and branches only after $\tau$.

    For every $t\in[0,T]$ we have $t-h\le T-h<\tau$, and therefore
    \[
        \alpha^{h,\mathcal D_{\mathrm{lag}}}(t,\omega_\theta)=\frac1h\mathcal D_{\mathrm{lag}}\big(R^{t,h}(\omega_\theta)\big)=\frac1hR(t-h,\omega_\theta)=0.
    \]
    Although $\mathcal D_{\mathrm{lag}}$ is not the zero map, it vanishes on all past windows generated by the present scenario set. Consequently,
    \[
        A^{h,\mathcal D_{\mathrm{lag}}}\equiv0, \qquad M^{h,\mathcal D_{\mathrm{lag}}}=R-R(0)=R.
    \]
    Every return path is continuously differentiable and of finite variation, so
    \[
        C^{h,\mathcal D_{\mathrm{lag}}}=[R,R]_\Pi\equiv0, \qquad O^{h,\mathcal D_{\mathrm{lag}}}\equiv0.
    \]
    The non-anticipative density convention therefore gives
    \[
        a^{h,\mathcal D_{\mathrm{lag}}}\equiv0, \qquad c^{h,\mathcal D_{\mathrm{lag}}}\equiv0, \qquad \nu^{h,\mathcal D_{\mathrm{lag}}}\equiv0, \qquad G^{h,\mathcal D_{\mathrm{lag}}}(T,\omega_\theta)=0.
    \]
    In particular, locally finite growth holds on every scenario; and, since the zero rule belongs to $\mathfrak W_\Omega^{h,\mathcal D_{\mathrm{lag}}}$, canonical num\'eraire admissibility holds on $\Omega$.

    Define the backward derivative functional
    \begin{equation*}
        D^-R(t,\omega):=
        \begin{cases}
            \qquad \qquad \qquad 0, & t=0,\\[0.3em]
            \displaystyle\lim_{\substack{r\downarrow0\\r\in\mathbb Q}}\frac{R(t,\omega)-R((t-r)\vee0,\omega)}{t-\big((t-r)\vee0\big)}, & t\in(0,T]\text{ and the limit exists},\\[1em]
            \qquad \qquad \qquad 0, & t\in(0,T]\text{ and the limit does not exist}.
        \end{cases}
    \end{equation*}
    For every rational $r>0$, the corresponding backward quotient, with value zero at $t=0$, is jointly Borel-measurable on $[0,T]\times\Omega$. The set on which the rational one-sided limit exists is Borel-measurable by the Cauchy criterion, and the zero convention on its complement therefore makes $D^-R$ jointly Borel-measurable. It is non-anticipative because every such quotient uses only values of the path observed no later than $t$. On the present scenario set, every return path is continuously differentiable and
    \[
        D^-R(t,\omega_\theta)=\dot R(t,\omega_\theta)=3\theta\big((t-\tau)^+\big)^2, \qquad (t,\theta)\in[0,T]\times\Theta_*.
    \]
    Along every scenario, this path is continuous and of finite variation.

    For $m\in\mathbb N$, set
    \begin{equation*}
        \pi^m(t,\omega):=mD^-R(t,\omega).
    \end{equation*}
    By Definitions~\ref{def:basic-admissible-follmer-integrands} and~\ref{def:scenario-wise-follmer-integrable}, these rules belong to $\mathfrak W_\Omega^{h,\mathcal D_{\mathrm{lag}}}$: for each fixed $\theta\in\Theta_*$ take the finite variation control $B_\theta(t):=D^-R(t,\omega_\theta)$ and the generating function $f_m(b,x):=mbx$, so that
    \[
        \partial_xf_m\big(B_\theta(t),R(t,\omega_\theta)\big)=mB_\theta(t)=\pi^m(t,\omega_\theta).
    \]
    Since each return path is continuously differentiable, the F\"ollmer integral agrees with the classical Riemann--Stieltjes integral. Using $D^-R=\dot R$, for every $\theta\in\Theta_*$ we obtain
    \begin{align*}
        R_{\pi^m}(T,\omega_\theta) = m\int_0^T\dot R(t,\omega_\theta)\,\d R(t,\omega_\theta) = m\int_0^T\big|\dot R(t,\omega_\theta)\big|^2\,\d t =\frac{9m\theta^2}{5}(T-\tau)^5.
    \end{align*}
    Also, $C_{\pi^m}^{h,\mathcal D_{\mathrm{lag}}}(\cdot,\omega_\theta)\equiv0$ for every $\theta\in\Theta_*$. Therefore, from \eqref{def:wealth},
    \begin{equation*}
        \inf_{\omega\in\Omega}X_{\pi^m}(T,\omega)\ge\exp\bigg(\frac{9m\underline\theta^2}{5}(T-\tau)^5\bigg)\xrightarrow{m\to\infty}\infty.
    \end{equation*}
    Thus, terminal wealth is not merely pointwise unbounded; the same sequence of non-anticipative portfolios becomes unbounded uniformly over the entire scenario set.

    Let $K(t,\omega):=\mathbf1_{\{T\}}(t)$ for $(t,\omega)\in[0,T]\times\Omega$, and put
    \[
        x_m:=\exp\bigg(-\frac{9m\underline\theta^2}{5}(T-\tau)^5\bigg).
    \]
    Then $x_mX_{\pi^m}(T,\omega)\ge1$, $\forall\,\omega\in\Omega$. Hence, for every $m$, the single non-anticipative rule $\pi^m$ finances the unit terminal withdrawal simultaneously on all scenarios from initial capital $x_m$. Since $x_m\downarrow0$, we have $x_\Omega(K)=0$ although $K\not\equiv0$, so the scenario market is not pathwise viable.

    The construction does not reveal the post-$\tau$ parameter in advance: all paths are indistinguishable up to and including time $\tau$, and the portfolio takes the same value there. After the branch, the rule reads the sign and magnitude from the already observed backward derivative, which on the present scenario set coincides with the current time derivative $\dot R$. This is the genuinely non-anticipative feature absent from a singleton scenario set.
\end{example}

\begin{remark}[The Reason the Two Layers Are Distinct]\label{rem:why-two-layers-distinct}
    The two preceding examples exhibit the same structural obstruction in complementary ways. The growth--num\'eraire layer is determined by the extracted finite-variation trend and the quadratic covariation. Therefore, it need not detect a nonconstant finite-variation component that remains in the residual path and has zero F\"ollmer quadratic variation. Such a component contributes neither extracted trend nor quadratic growth penalty, yet suitable wealth-generating F\"ollmer portfolios may exploit its directional variation.

    Example~\ref{ex:lfg-without-pointwise-boundedness} demonstrates this phenomenon for a concrete last-cell Faber--Schauder extractor in the simplest singleton market. Example~\ref{ex:multi-scenario-non-anticipative-separation} shows that the phenomenon is not an artifact of the advance knowledge implicit in a singleton scenario set. There, uncountably many paths agree up to a common branching time; while one sequence of jointly Borel-measurable, non-anticipative portfolio rules, using only backward derivatives already observed, generates terminal wealth that diverges uniformly over all scenarios.

    By contrast, in the continuous semimartingale setting of \cite{KaratzasKardaras2021}, the residual in the canonical decomposition is a continuous local martingale. A continuous local martingale with zero quadratic variation is constant; moreover, when the structural condition for the num\'eraire holds, relative wealth is a nonnegative local martingale and hence a supermartingale. These properties rule out hidden finite-variation directional gains and provide the budget inequalities that connect the growth--num\'eraire layer to the viability--boundedness layer. No such connection follows automatically from a general pathwise residual and its F\"ollmer quadratic covariation.
\end{remark}

\newpage

\noindent \textbf{Funding}

\medskip

\noindent Ioannis Karatzas gratefully acknowledges support from the National Science Foundation under Grant DMS-25-06199, and from a Lenfest Award at Columbia University. Donghan Kim gratefully acknowledges support from the National Research Foundation of Korea under Grant RS-2025-00513609, funded by the Korean government (MSIT), and from an EWon Assistant Professorship Award at KAIST. 

\bigskip

\renewcommand{\refname}{References}
\bibliography{pathwise1}
\bibliographystyle{apalike}

\end{document}